\documentclass[11pt,onecolumn,letterpaper]{article}
\usepackage{microtype}
\usepackage{paralist}
\usepackage{amsthm}
\usepackage{fullpage}
\usepackage{amsmath,amssymb}
\usepackage{complexity}
\usepackage{xspace}
\usepackage{bm}
\usepackage{soul}
\usepackage{tabularx}
\usepackage{graphicx}
\usepackage{lmodern}
\usepackage{framed}
\usepackage{algorithm}
\usepackage{array}
\usepackage{color}
\usepackage[noend]{algpseudocode}
\usepackage{enumitem}
\usepackage{multirow}
\usepackage{url}

\date{}

\newcommand{\para}[1]{\vspace{0.2em}\noindent\textbf{#1.}~}

\makeatletter
\newtheorem*{rep@theorem}{\rep@title}
\newcommand{\newreptheorem}[2]{%
\newenvironment{rep#1}[1]{%
 \def\rep@title{#2 \ref{##1}}%
 \begin{rep@theorem}}%
 {\end{rep@theorem}}}
\makeatother

\renewcommand{\le}{\leqslant}
\renewcommand{\ge}{\geqslant}
\renewcommand{\geq}{\geqslant}
\renewcommand{\leq}{\leqslant}

\newcommand{\disjoint}{\mathsf{Disj}\xspace}

\newcommand{\msgs}{\mathsf{Msgs}\xspace}
\newcommand{\GS}{\mathcal{GS}\xspace}
\newcommand{\CG}{\mathcal{CG}\xspace}

\def\etc{etc.\@\xspace}

\newtheorem{theorem}{Theorem}
\newtheorem{observation}[theorem]{Observation}
\newtheorem{lemma}[theorem]{Lemma}
\newtheorem{corollary}[theorem]{Corollary}
\newtheorem{claim}[theorem]{Claim}

\newreptheorem{lemma}{Lemma}
\newreptheorem{theorem}{Theorem}

\let\originalleft\left
\let\originalright\right
\renewcommand{\left}{\mathopen{}\mathclose\bgroup\originalleft}
\renewcommand{\right}{\aftergroup\egroup\originalright}
\newcommand{\prob}{\mathrm{Pr}\xspace}
\newcommand{\Prob}[1]{\mathrm{Pr}\left[#1\right]\xspace}

\newcommand{\expect}[1]{\mathbb{E}\left[#1\right]\xspace}

\newcommand{\vol}{\mathsf{Vol}\xspace}

\newenvironment{proofoutline}{%
  \proof}{\endproof}

\newlength\myindent
\begin{document}
\title{Leader Election in Well-Connected Graphs\thanks{Authors are listed alphabetically.}}

\author{
Seth Gilbert \\
{\em National University of Singapore}\\
\url{seth.gilbert@comp.nus.edu.sg}
\and
Peter Robinson\\
{\em McMaster University, Canada}\\
\url{peter.robinson@mcmaster.ca}
\and
Suman Sourav\thanks{This author is the corresponding author.} \\
{\em National University of Singapore}\\
\url{sourav@comp.nus.edu.sg}
}
\maketitle

\begin{abstract}
In this paper, we look at the problem of randomized leader election in synchronous distributed networks with a special focus on the message complexity. We provide an algorithm  that solves the implicit version of leader election (where non-leader nodes  need not be aware of the identity of the leader) in any general network with $O(\sqrt{n} \log^{7/2} n \cdot t_{mix})$ messages and in $O(t_{mix}\log^2 n)$ time, where $n$ is the number of nodes and $t_{mix}$ refers to the mixing time of a random walk in the network graph $G$. For several classes of well-connected networks (that have a large conductance or alternatively small mixing times e.g. expanders, hypercubes, etc), the above result implies extremely efficient (sublinear running time and messages) leader election algorithms. 
Correspondingly, we show that any substantial improvement is not possible over our algorithm, by presenting an almost matching lower bound for randomized leader election. We show that $\Omega(\sqrt{n}/\phi^{3/4})$ messages are needed for any leader election algorithm that succeeds with probability at least $1-o(1)$, where $\phi$ refers to the conductance of a graph. 
To the best of our knowledge, this is the first work that shows a dependence between the time and message complexity to solve leader election and the connectivity of the graph $G$, which is often characterized by the graph's conductance $\phi$. Apart from the $\Omega(m)$ bound in \cite{Kutten:2015:CUL:2742144.2699440} (where $m$ denotes the number of edges of the graph), this work also provides one of the first non-trivial lower bounds for leader election in general networks. 
\end{abstract}
\newpage
\pagenumbering{arabic}

\section{Introduction} \label{sec:intro}

Leader election is one of the most classical and fundamental problem in the field of distributed computing having applications in numerous problems relating to synchronization, resource allocation, reliable replication, load balancing, job scheduling (in master slave environment), crash recovery, membership maintenance \etc Computing a leader can be thought of as a form of symmetry breaking, where exactly one special node or process (denoted as leader) is chosen to take some critical decisions. %

Loosely speaking, the problem of leader election requires a set of nodes in a distributed network to elect a unique leader
among themselves, i.e., exactly one node must output the decision that it is the leader. There are two well known variants of this problem (cf. \cite{AW98,Lynch:1996:DA:525656}), the explicit variant where at the end of the election process all the nodes are required to be aware of the identity of the leader and the implicit variant where the non-leader nodes need not be aware of the identity of the leader. %

Often, the implicit variant is sufficient for many practical applications, e.g. its original application for token generation in token ring environments \cite{DBLP:conf/ifip/Lann77} \etc This variant also allows us to clearly distinguish between the two aspects of explicit leader election and costs associated to each of them, i.e. electing a leader (implicitly) as compared to broadcasting the unique id of the leader to all the other nodes. Clearly, any solution for the explicit variant of leader election also solves the implicit variant. However, it is to be noted that any solution for the implicit leader election could solve explicit leader election by broadcasting the identity of the leader to all the nodes. 
In this paper, we mainly focus on the implicit variant of leader election on a network without edge or link failures.

Compared to deterministic solutions that provide absolute guarantees for the election of a
unique leader, randomized solutions guarantees unique leader election with high probability.
However, this weakened assumption is still sufficient for many practical applications. With an acceptable error probability, randomization can result in a significant reduction in time and message complexities. This is highly advantageous for large scale distributed systems (e.g. P2P systems, overlay and sensor networks \cite{Ratnasamy:2001:SCN:964723.383072,Rowstron:2001:PSD:646591.697650,Zhao:2006:TRG:2312202.2314945}), where scalability is an important issue. Furthermore, in anonymous networks, a randomized solution is often possible by randomly assigning unique identifiers to nodes (as done herein) whereas a corresponding deterministic solution is impossible (see \cite{Angluin:1980:LGP:800141.804655}).

This paper focuses on studying the message complexity of implicit leader election in synchronous distributed networks. Here, we show the relationship between the graph connectivity (which is characterized by the graph's conductance $\phi$) with the time and number of messages required for leader election. We provide an algorithm that solves implicit leader election in any general network with $\tilde{O}(\sqrt{n}\cdot t_{mix})$ messages and in $\tilde{O}(t_{mix})$ time, where $n$ is the number of nodes and $t_{mix}$ refers to the mixing time of a random walk in the network graph $G$.\footnote{Throughout the paper, we use the $\tilde{O}$ notation to hide $\log n$ factors.} Correspondingly, we show that $\Omega(\sqrt{n}/(\phi)^{3/4})$ messages are needed for any leader election algorithm that succeeds with high probability, where $\phi$ refers to the graph's conductance. We also show that the knowledge of the network size $n$ is critical to achieve the said message and time complexities, but surprisingly, knowledge of other graph properties, such as the conductance, mixing time, or diameter is not needed.

\para{Computing Model}
We model the network as a connected, undirected graph $G = (V,E)$ with $|V|=n$ nodes and $|E|=m$ edges where nodes communicate over the graph edges.
We assume synchronous communication that follows the standard $\mathcal{CONGEST}$ model \cite{peleg}. %
In each round, each node can perform some local computation which can involve accessing a private source
of randomness. Additionally, each node $u$ is allowed to send a message of size $O(\log n)$ bits through each edge $(u,v)$ incident on $u$. %
Nodes do not have predefined ids and there are no node or link failures. 

\para{Port Numbering Model}
We assume that the nodes know the network size $n$ and wake up simultaneously at the beginning of the execution. %
Also, nodes are anonymous in the sense that they do not have unique IDs.\footnote{Our lower bound holds even if nodes start with unique IDs.} Each node chooses an id uniformly at random from within the range $[1,n^4]$.\footnote{This range guarantees that the chosen values are unique with high probability \cite{Lynch:1996:DA:525656} (Chapter 4, page 72).}
Each node $u$ with degree $|d_u|$ has ports $1,\dots,d_u$, over which it can send messages across undirected links to its  neighbors; that is, each neighbor  of $u$ is the endpoint of exactly $1$ of $u$'s ports. Nodes only know the port numbers of connections and are unaware of their neighbors' identities.
We do not assume these mappings to be symmetric: in particular, it can happen that $u$ is connected to $v$ via port number $i$, whereas $v$ is connected to $u$ via port number $j \ne i$.

\para{Randomized Implicit Leader Election}
Every node of a given distributed network has a flag variable (or a boolean variable) initialized to $0$ and, after the process of election, with high probability (w.h.p), only one node, called the leader, raises its flag by setting the flag variable to $1$. An algorithm $\mathcal{A}$ is said to solve leader election in $T$ rounds, if within $T$ rounds nodes elect a unique leader w.h.p., and none of the nodes send any more messages after $T$ rounds. %

\para{Results}
In this paper, we present both upper and lower bounds for the problem of implicit leader election in general networks. We provide an algorithm that solves implicit leader election in any general network with $O(\sqrt{n} \log^{7/2} n \cdot t_{mix})$ messages and in $O(t_{mix}\log^2 n)$ time, where $n$ is the number of nodes and $t_{mix}$ refers to the mixing time of a random walk in the network graph $G$. If larger message sizes of $O(\log^3 n)$ is allowed, the message and time complexity  reduces to $O(\sqrt{n} \log^{3/2} n \cdot t_{mix})$ and $O(t_{mix})$ respectively.  This implies significantly faster and efficient solutions (that are sub-linear in terms of message complexity) for leader election in several important classes of well-connected graphs that have a large conductance or alternatively a small mixing time. For example, to solve implicit leader election, in expander graphs (see \cite{hoory06} for applications) that have a mixing time of $O(\log n)$, it takes only $O(\log^3 n)$ time and $O(\sqrt{n} \log^{9/2} n)$ messages; in hypercube  graphs, that have a mixing time of $O(\log n\log \log n)$, it takes only $O(\log^3 n\log \log n)$ time and $O(\sqrt{n} \log^{9/2} n\log \log n)$ messages. The algorithm can also be used for solving the explicit variant of leader election by adding a broadcasting procedure, wherein the leader broadcasts its identity to all other nodes. For well connected graphs, this breaks the $\Omega(m)$ lower bound given in \cite{Kutten:2015:CUL:2742144.2699440} (where $m$ denotes the number of edges of the graph) and nearly matches the $\Omega(\sqrt{n})$ lower bound for clique graphs  \cite{Kutten:2015:SBR:2945722.2945791} (as cliques have constant conductance). 

We show that a dependence on the graph conductance is unavoidable, by presenting a message complexity lower bound of $\Omega(\sqrt{n}/\phi^{3/4})$ messages that holds for any leader election algorithm that succeeds with probability at least $1-o(1)$. This nearly matches the upper bound since we know that $\Theta(1/\phi) \le t_{mix} \le \Theta(1/{\phi}^2)$ from \cite{Sinclair}.

By a similar analysis, we also provide lower bounds for other graph problems like broadcast and spanning tree construction in terms of the graph's conductance. Our lower bounds also apply for the $\mathcal{LOCAL}$ model \cite{peleg}, where there are no restrictions on the message size. Other than the $\Omega(m)$ bound in \cite{Kutten:2015:CUL:2742144.2699440}, to the best of our knowledge, this is the first non-trivial lower bound for randomized leader election in general networks. Also, ours is one of the first results to show the dependence of the time and message complexity to solve leader election on the connectivity of the graph $G$, which is often characterized by the graph's conductance $\phi$.

Additionally, we show that the knowledge of the network size $n$ is critical for our algorithm to succeed by giving a lower bound of $\Omega(m)$ for all graphs if $n$ is not known. However surprisingly, the knowledge of other graph properties, like the conductance, mixing time, or diameter is not needed.

\para{Prior Works}
Leader election, being one of the fundamental paradigms in the theory of distributed computing has been widely studied.
The problem was first stated by G{\'{e}}rard Le Lann in \cite{DBLP:conf/ifip/Lann77} in the context of token ring networks, and thereafter has been extensively studied for various types of networks, scenarios and communication models. For particular types of networks topologies like token rings, mesh, torus, hypercubes and cliques the problem of leader election has been well-studied resulting in specialized algorithms and lower bounds in terms of both, time and message complexities (e.g. \cite{Chang:1979:IAD:359104.359108,DOLEV1982245,Vitanyi:1984:DEA:800057.808725,peleg,Lynch:1996:DA:525656,FLOCCHINI199676,Afek:1985:TMB:323596.323613, doi:10.1137/0216019,Kutten:2015:SBR:2945722.2945791,Ramanathan2007} and references therein). In a seminal paper Gallager, Humblet and Spira \cite{Gallager:1983:DAM:357195.357200}, provided a deterministic solution for general graphs by finding the minimum spanning tree of the graph in $O(n\log n)$ time and exchanging a total of $O(m\log n)$ messages. Thereafter, Awerbuch \cite{Awerbuch:1987:ODA:28395.28421} provided an $O(n)$ round deterministic algorithm with a message complexity of $O(m + n \log n)$ messages, where $m$ refers to the total number of edges in the graph. %
Peleg \cite{Peleg:1990:TLE:78028.78040} provided an $O(D)$ time optimal algorithm for leader election with a message complexity of $O(mD)$, where $D$ the diameter of the graph. More recently, in \cite{Kutten:2015:CUL:2742144.2699440}, the authors provide an algorithm that requires only $O(m)$ messages though it could take arbitrary (albeit finite) time. 
There also exists significant amount of literature (see \cite{10.1007/3-540-40026-5_6, 815321, 491576, Brunekreef1996} and references therein) that provides a solution for leader election on fault prone networks, with possible node or link failures. 

The best known bounds for general graphs are as follows. In \cite{Kutten:2015:CUL:2742144.2699440} Kutten et al. show that $\Omega(m)$ is the lower bound on messages and $\Omega(D)$ is the lower bound on time for any implicit leader election algorithm. They compare and contrast the deterministic algorithms to randomized algorithms while trying to simultaneously achieve optimal time and message complexity for leader election. Unlike the deterministic case where an algorithm cannot be simultaneously time and message optimal (e.g. in a cycle any $O(n)$ time deterministic algorithm requires at least $\Omega(n \log n)$ messages even when nodes know $n$ \cite{Frederickson:1987:ELS:7531.7919}), they show that for the randomized case simultaneous optimality can be achieved in certain cases. In particular, to show the bounds are tight they give an algorithm that takes $O(m)$ messages (not time optimal), an algorithm that takes $O(m\log \log n)$ messages and $O(D)$ time (almost simultaneously optimal).

In \cite{Kutten:2015:SBR:2945722.2945791}, Kutten et al. show that in terms of message complexity, there exists a gap between the implicit and the explicit variants of leader election. For the explicit variant, all nodes needs to be informed of the identity of the leader, and as such $\Omega(n)$ messages is the obvious lower bound for all graphs. However, for the implicit version the authors  by provide a sub-linear bound algorithm on complete networks that runs in $O(1)$ rounds and (w.h.p.) uses only $O(\sqrt{n}\log^{3/2}n)$ messages to elect a unique leader (w.h.p.). Thereafter, they extend this algorithm to solve leader election on any connected graph $G$ in $O(t_{mix})$ time and $O(t_{mix} \cdot \sqrt{n}\log^{3/2}n)$ messages, where $t_{mix}$ is the mixing time of a random walk on $G$. 

A key difference, however, was that they assume that every node in the graph knows the mixing time of the graph, which significantly simplifies the problem.  An important challenge addressed by our algorithm is (effectively) estimating when a collection of random walks is well-enough mixed. While there is a recent result in \cite{Molla:2017:DCM:3007748.3007784} that shows how nodes can quickly estimate the mixing time of the graph, their algorithm requires $\Omega(m)$ messages and hence cannot be used for the purpose of achieving a small message complexity, where $m$ is the total number of edges in the graph. 

In the context of using random sampling for leader election, the work of \cite{byz-leader} uses random walks to limit the impact of Byzantine nodes on electing an honest leader in dynamic networks.

\section{Preliminaries}

In this section, we describe some basic definitions and concepts that we make use of throughout the paper. First, we give a brief overview and the definition of graph conductance. Next, we describe some basic notation for random walks on a graph $G$ including its mixing time, and state the relationship between the mixing time and the conductance of $G$.

Conductance, in general, is a characterization of the bottleneck in communication of a graph.
The notion of graph conductance was introduced by Sinclair \cite{Jerrum:1988:CRM:62212.62234}.
For a given graph $G = (V,E)$, a subset of nodes $U \subseteq V$ and cut $\mathcal{K}=(U, V\setminus U)$,  we define $E_{\mathcal{K}}$ to be the subset of edges across the cut $\mathcal{K}$, and the volume $\vol(U) = \sum_{v\in U}d_v$, where $d_v$ refers to the degree of node $v$.
The cut-conductance is defined as $ \phi_{\mathcal{K}} = {|E_{\mathcal{K}}|}/{\min\{\vol(U),\vol(V\setminus U)\}}$. 
The conductance of the graph $G$ is defined as the minimum of the cut-conductance across all possible cuts $\tilde{\mathcal{K}}$, i.e., $\phi(G) = \min \{ \phi_{\mathcal{K}} \mid {\mathcal{K} \in \tilde{\mathcal{K}}} \}$. %
We simply write $\phi$  instead of $\phi(G)$, when graph $G$ is clear from the context.

For a random walk on $G(V,E)$, we define a node set $V = \{ v_1, \dots , v_n \}$ and an $n \times n$ transition matrix $P$ of $G$. Each position of the form $P[i,i]$ in the transition matrix has an entry $p_{v_i,v_i} =1/2$, else if $i \neq j$ then $P[i,j]$ has an entry $p_{v_i,v_j} = 1 / 2d_{v_i}$  if there is an edge $(v_i , v_j) \in E$, otherwise $p_{v_i,v_j} = 0$. At a particular step, the entry $p_{v_i,v_j}$ gives the probability of a random walk moving from node $v_i$ to node $v_j$. This exactly corresponds to a \textit{lazy} random walk wherein a random walk either stays in the current node  with probability $1/2$; otherwise moves to a neighbor with probability $1/2d_{v_i}$. The probability distribution $\pi_t$ determined by $P$ represents the position of a random walk after $t$ steps. If some node $v_i$ starts a random walk, the initial distribution $\pi_0$ of the walk is an $n$-dimensional vector having all zeros except at index $i$ where it is $1$. After the node $v_i$ has chosen to forward the random walk token, either to itself or to a random neighbor, the distribution of the walk (after $1$ step) is given by $\pi_1 = P\pi_0$ and in general we have $\pi_t = P^t\pi_0$. For any %
connected graph $G$, the distribution will eventually converge to the stationary distribution $\pi_*=(q_1,\dots,q_n)$, which has entries $q_i = d_{v_i}/(2|E|)$ and satisfies $\pi_* = P \pi_*$.
The mixing time of an $n$-node graph $G$, $t_{mix}(G)$ is defined as the minimum $t$ such that, for each starting distribution $\pi_0$, $\parallel P\pi_t-\pi_*\parallel_\infty \hspace{1mm} \le \frac{1}{2n}$, where $\parallel \cdot \parallel_\infty$ denotes the usual maximum norm on a vector. We simply write $t_{mix}$  instead of $t_{mix}(G)$, when graph $G$ is clear from the context. %

Note that, the connectedness of a graph $G$ (determined by its conductance $\phi$) and the mixing time of a random walk on $G$ are closely related. Better connectivity implies fast mixing and vice versa. 
There is a well known result that formally relates the graph conductance $\phi$ to the mixing time (see \cite{Sinclair}) as follows
\begin{equation}\label{eq:phimix}
\Theta(1/\phi) \le t_{mix} \le \Theta(1/{\phi}^2)
\end{equation}

\section{A Leader Election Algorithm for Well-Connected Networks} \label{sec: algorithm}

In this section, we provide an algorithm that solves implicit leader election for any given graph $G$ with time complexity of $\tilde O(t_{mix})$  and more importantly, with $\tilde O(\sqrt{n} t_{mix})$ message complexity, where $t_{mix}$ refers to the mixing time of a random walk on the graph $G$. 

The proposed algorithm, in its initial phase is similar to the algorithms given in \cite{Gilbert:2010:DAO:1873601.1873679, Kutten:2015:SBR:2945722.2945791} where the initial objective is to reduce the number of competing nodes (contenders), while also ensuring that there exists at least one contender (w.h.p.). For this purpose, each node $v$ in the network graph $G$, elects itself as a contender with a probability of $c_1\log (n)/ n$, where %
$c_1$ is a sufficiently large constant. As such, the probability of no node electing itself as a contender is $(1-\frac{c_1\log n}{n})^n \approx \exp (-c_1\log n) = n^{-c_1}$ ; thus, implying that w.h.p. the number of contenders is nonzero. %

Now, imagine a scenario in which each of these contenders contacts a set of nodes, which we refer to as the contender's target set. If the target set is large enough (say, $n/2 +1$), then for any two contenders we can say that there would be common/intersecting node that would have communicated with both contenders. Thereafter, the contenders can communicate via this intersecting node. If all contenders have a sufficiently large target set then all contenders would be able to communicate with one another. We design our algorithm based on this idea.

First, we determine the minimum size of the target set needed to guarantee an intersection w.h.p. between the target sets of any two contenders. It can be easily shown with the birthday paradox argument that if any two contender nodes $u$ and $v$ contact $O(\sqrt{n\log n})$ random nodes, then w.h.p. there is at least one node $w$ that was chosen by both $u$ and $v$. %
By the definition of mixing time, if a random walk has taken at least $t_{mix}$ steps, then (for all practical purposes) its end point can be considered as a random node. Therefore, each contender node can find $O(\sqrt{n\log n})$ random nodes by performing $O(\sqrt{n\log n})$ independent random walks in parallel. The random walks essentially function as mechanisms for selecting/sampling ``random" nodes, where the guarantee is if the length of the walk is long enough, the choice is close to uniform. We might as well think of the random walks as a black box that return a collection of random nodes. However, as nodes are not aware of the mixing time of the graph, this technique cannot be used directly to obtain random nodes.  %

Without knowledge of the mixing time, it is difficult (if not impossible) to obtain a set of nodes that are chosen uniformly at random by using random walks. Therefore, the major challenge reduces to correctly obtaining a set of possibly non-random nodes (as random walks might be of length less than the mixing time) that satisfy the required properties that we had hoped to achieve from a uniformly random chosen target set.
One difficulty here lies is in determining the ideal length of the random walks of each contender without the knowledge of the mixing time.
To deal with this, in our algorithm, we use a guess and double strategy where in each iteration nodes guess a length for the random walk, perform random walks of the chosen length, determine based on some criteria if the length is sufficient; if not the next iteration begins with double the previous estimate. 

The critical part is to determine the criterion for which we can consider the length of the random walks to be sufficient. A natural solution would be to check if there are enough intersections in the target set (with target sets of other contenders). For example, if the target set of each contender had an intersection with target sets of all other contenders, all of them can communicate via the intersecting node(s). However, then we would require the knowledge of the exact number of contenders to determine termination, which is difficult to obtain with certainty. In fact, we show that an intersection with greater than half of the contenders is sufficient and obtainable. 

Given such a criterion, it creates another challenge, as it might be the case that all the random walks do not terminate in the same round. For example, consider the case where a large number of contenders belong to the same locality of the graph and as a result they contact each other quickly, via their random walks. However, a few of the contenders do not belong to this neighborhood and are slightly far off from this locality. In this case, the target set of the locally placed contenders would belong to the same locality (and not be nearly randomly spread). As such, it would be difficult for the far flung contenders to make contact with any of the locally placed contender's target sets, requiring much longer lengths of random walks than the mixing time. 
For this case, we would also need to guarantee that the random walks that terminate early are still easily discoverable.

To deal with the above challenges, we provide a twofold stopping criterion: first, we want to ensure that the end points of the random walks of a contender intersect with the random walks of at least half of the total number of contenders; second, we would also like to ensure that the end points of these random walks are sufficiently spread out, such that other random walks do not spend too much time discovering them. %
Another crucial part to consider is dealing with the congestion that might be caused by the information carried along the random walks.

Basically, the given randomized leader election algorithm can be divided into three major parts. First, a node makes a probabilistic decision determining its candidature, i.e., whether or not it becomes a contender. Then, in the second part, contenders guess and double length of random walks until it satisfies some required properties. Lastly, based on information retrieved from random walks, a node elects itself as the leader if it satisfies a certain winning condition.

We provide the following contender lemma which restricts the total number of possible contenders.
\begin{lemma} (Contender Lemma)\label{lem:contenders}
  With high probability the number of contenders is in the range  $[ \frac{3}{4}c_1 \log n, \frac{5}{4} c_1 \log n]$, where $c_1$ is a sufficiently large constant.
\end{lemma}
\begin{proof}
Since each node becomes contender independently with probability $(c_1\log n)/{n}$, we can apply two tail bounds to show concentration around the expected number of contenders ${c_1\log n}$.
Let $X$ be the number of contenders.
By standard Chernoff Bounds (Theorems 4.4 and 4.5 in \cite{Mitzenmacher:2005:PCR:1076315}), we know that
$
  \Prob{X \ge \left(1+{1}/{4}\right)c_1\log n} \le \exp\left(-\left({1}/{4}\right)^2 c_1\log(n) / 3\right)
$
and, similarly,
$
  \Prob{X \le \left(1 - {1}/{4}\right)c_1\log n} \le \exp\left(-\left({1}/{4}\right)^2 c_1\log(n) / 2\right).
$
For sufficiently large $c_1$, both of these bounds can be shown to hold with high probability and hence the lemma follows by a simple union bound.
\end{proof}

Each contender node $u$ creates  $c_2 \sqrt{n \log n}$ tokens and starts $c_2 \sqrt{n \log n}$ random walks of length $t_u$ in parallel, where $c_2$ is a constant $> 2$. Each random walk is represented by a token $\langle u, t_u \rangle$ (of $O ( \log n )$ bits), where $u$ represents the node's id, and $t_u$ represents the length of the random walk. At each step of the random walk $t_u$ is decremented by 1, until it finally becomes 0. We define \textit{proxies} of node $u$ as the nodes where the random walks generated by $u$ complete $t_u$ steps, where $t_u$ is either $u$'s current or final guess of the length of the random walk. Two contender nodes are said to be \emph{adjacent} if they share at least one proxy.  

\noindent The algorithm guarantees the following properties at the end of the execution: \vspace{-2mm}
\begin{itemize}
\item \emph{Intersection Property}: A contender $u$ satisfies the intersection property iff $u$ is adjacent to at least $\geq \frac{3}{4}c_1 \log n$ of the other contender nodes. Using Lemma \ref{lem:contenders}, we see that any node which satisfies the intersection property is adjacent to greater than half of the total number of contenders (as $\frac{3}{4}c_1 \log n \geq \frac{1}{2}\left( \frac{5}{4} c_1 \log n\right)$) w.h.p. \vspace{-2mm}
\item \emph{Distinctness Property}: A contender node $u$ satisfies the distinctness property if $\geq \tfrac{c_{2}}{2} \sqrt{n \log n}$ of its proxies are distinct. For a particular guess of $t_u$, a proxy $p_u$ of a random walk belonging to $u$ is called a distinct proxy only if $p_u$ is the end node of exactly one random walk belonging to $u$ (from among the $c_2\sqrt{n\log n}$ many random walks belonging to $u$) i.e. no other random walks belonging to $u$ ends at $p_u$. For any contender node, the spreading out of its random walks is characterized by the number of distinct proxies. \vspace{-2mm}
\end{itemize}

In the $\mathcal{CONGEST}$ model, due to the restriction on the message size,  it is impossible to perform too many walks in parallel along an edge. We solve this issue by sending only one token and the count of tokens that need to be sent by a particular contender which is still $O(\log n)$, and not all the tokens themselves.
Similarly, our algorithm also requires some id information (set of ids of other contenders) to be sent along the random walk. %
We note that the maximum possible number of contenders is $\le \frac{5}{4} c_1 \log n$ w.h.p. (c.f. Lemma \ref{lem:contenders}). In the worst case, an intermediate node might have to deal with $O(\log n)$ messages of $O(\log^2 n)$ size each, introducing a maximum possible delay of $O(\log^2 n)$ rounds. 
To account for this delay, in the algorithm, we define $T = \frac{25}{16} c_1 t_u \log^2 n  $, and use this upper bound estimate to keep the execution of the algorithm in synchrony. We relegate the formal details of handling congestion to the proof of Lemma~\ref{lem:msgcomplexity}.

\begin{algorithm}
\caption{Leader Election: Initialization}
\label{algo:leader1}
\begin{algorithmic}[1]
\State Each node generates a random id in the range $[1, \ldots, n^4]$.
\State Each node designates itself a contender with probability $c_1(\log{n}/n)$.
\State Each contender begins the protocol by executing a Random Walk Phase of length $O(1)$.
\State Any node that is not a contender declares itself as non-leader. 
\end{algorithmic}
\end{algorithm}

\begin{algorithm}
\caption{Leader Election: Random Walk Phase of length $t_u$ of contender $u$.}
\label{algo:leader2}
\begin{algorithmic}[1]
\State Each contender $u$ initiates $c_2 \sqrt{n \log{n}}$ parallel random walks of length $t_u$ for time $T = O(t_u \log^2{n})$.
\State When a random walk completes, the last node in the random walk is called a \textbf{proxy of $u$}.
\State Node $u$ then performs three synchronized rounds of information exchange with its proxies, each taking time $T = O(t_u \log^2{n})$:
\Statex \textbf{Round 1.} Each proxy sends back its id, a Boolean $d$ determined by distinctness, and the set $I_1$, which contains the ids of the other contenders for which it is also a proxy.
\Statex \textbf{Round 2.} $u$ sends set $I_2$ to its proxies, which is the union of the $I_1$ sets received in round~1.
\Statex \textbf{Round 3.} Proxies send back set $I_3$: the union of the $I_2$ sets received.
\State Contender $u$ decides to stop if the \emph{Intersection Property} and the \emph{Distinctness Property} are met for the set $I_2$.  %
\State Let $I_4$ be the union of all $I_3$ sets received by $u$. If $u$ decides to stop,  has the largest id in set $I_4$, and it has not previously received any \emph{winner} messages, then it designates itself as the leader and sends its proxies a \emph{winner} message. %
\State The first time a proxy receives a \emph{winner} message, it sends it to all its contenders.
\State The first time a contender receives a \emph{winner} message, it sends it to all its proxies and appends it to all future messages.
\State At the end of the random walk phase, a contender that has not decided to stop waits $2T$ time (for \emph{winner} messages to propagate) and then begins a new Random Walk Phase of length $2t_u$.
\State Any contender that has stopped and is not a leader, declares itself as non-leader.
\end{algorithmic}
\end{algorithm}

Consider contender nodes that are yet to satisfy the intersection and distinctness properties as active nodes; consequently nodes that have already satisfied the said properties are considered inactive. That is, all nodes that will not double their estimate of $t_u$ are considered as inactive. It is to be noted that all the active nodes are synchronous and for all inactive nodes, the distance to their respective proxies is less than the current estimate of $T$ (of the active nodes).

\begin{observation} \label{obv:inactive}
All inactive contender nodes satisfy both the intersection and the distinctness properties. %
\end{observation}

\begin{lemma} \label{lem:active}
For any active contender node $y$, after the iteration where $t_y = c_3t_{mix}$ $(c_3 \geq 1)$, w.h.p. $y$ satisfies both the intersection and distinctness properties. In fact, $y$ has intersecting proxies with all of the other contenders (both active and inactive). %
\end{lemma}

\begin{proof}

Consider a set $Y$ consisting of all active contenders and a set $X$ of all the inactive contenders (contenders that decide not to double their estimate after some previous epoch). We prove the lemma using the following claims. %

\begin{claim}
Each contender node in $Y$ is adjacent to (has intersecting proxies with) all the other contender nodes, w.h.p. %
\end{claim}
\begin{proof}
For a contender node $y \in Y$, when $t_y=c_3t_{mix}$ $(c_3 \geq 1)$, $y$ has $c_2 \sqrt{n \log n}$ random proxies by running $c_2 \sqrt{n \log n}$ independent random walks of length $= c_3t_{mix}$ (proxies are random by the definition of mixing time). For any contender node $x \in X$, since $x$ satisfies the distinctness property, it has at least $\frac{c_{2}}{2} \sqrt{n \log n}$ distinct proxies. The probability of non-intersection between this set of $\frac{c_{2}}{2} \sqrt{n \log n}$ distinct proxies and the set $c_2 \sqrt{n \log n}$ random proxies is given by a birthday-paradox style argument to be $(1-\frac{(c_2/2)\sqrt{n\log n}}{n})^{c_2\sqrt{n\log n}} = \exp(-\frac{(c_2)^2}{2}\log n) = O(\frac{1}{n})$. The statement is true for all pair of nodes by taking a simple union bound.
Thus, with high probability, each contender node in $Y$ has at least one common/intersecting proxy with any contender node in $X$.

Now, consider two different contenders $y_1,y_2 \in Y$, each of which has a set of $c_2 \sqrt{n \log n}$ random proxies. 
Using similar arguments as above it can be easily shown that each contender node in $Y$ has at least one common/intersecting proxy with every other contender node in $Y$, w.h.p. 

This implies that each contender in $Y$ has intersecting proxies with all the other contenders, both in $X$ and $Y$, and thus is adjacent to all the other contenders.
\end{proof}

\begin{claim}
Each contender node $y$ in $Y$ satisfies the distinctness property,  %
w.h.p., when $t_y=c_3t_{mix}$, where $c_3$ is a constant $\geq 1$.
\end{claim}
\begin{proof}
To show the number of distinct proxies, we name the $c_2 \sqrt{n \log n}$ independent random walks of any contender $y \in Y$ as $w_1, w_2, \dots, w_{c_2 \sqrt{n \log n}}$. After the random walks have taken $c_3t_{mix}$ $(c_3 \geq 1)$ number of steps, the probability that two of these random walks $w_i$ and $w_j$ do not share a proxy  is $\approx 1-\tfrac{1}{n}$ (by the definition of mixing time). The probability that no other the random walks of node $u$ ends up at the same node as $w_i$ is obtained by taking an union bound. \vspace{-3mm}
\begin{equation}
\Prob{w_i \text{ has a distinct proxy}} \ge \left(1-\frac{1}{n}\right)^{c_2 \sqrt{n \log n}} \ge \exp \left(-c_2\sqrt{\frac{\log n}{n}}\right) \geq 1- c_2\sqrt{\frac{\log n}{n}}
\end{equation}
{The above equation holds as $\left(1-\frac{1}{x}\right) = \exp(-1)$ and $ \exp(x) \geq 1+x$.}

Let $X$ be a binary random variable such that $X_i=0$ when $w_i$ has a distinct proxy (no other $w_j$ ends at the proxy of $w_i$), and $X_i=1$ when it does not. 
This implies (from above) that $\Prob {X_i=0} \geq 1- c_2\sqrt{\frac{\log n}{n}}$ and $\Prob{X_i=1} \leq  c_2\sqrt{\frac{\log n}{n}}$. 
We define another binary random variable $Y$ such that $Y_i = 1$ with probability $c_2\sqrt{\frac{\log n}{n}}$, otherwise $0$. Clearly, $\Prob{Y_i =1}$ is always $\ge \Prob{X_i=1}$.
Since each $Y_i$ is independent of one another,  %

\[
	\expect {\sum_{i=0}^{c_2 \sqrt{n \log n}}Y_i} = c_2\sqrt{\frac{\log n}{n}} \cdot c_2 \sqrt{n \log n} = {c_2}^2\log n
\]
Thereafter, using a standard chernoff's bound we show a bound on the summation over $Y_i$. %
\[
	\Prob {\sum_{i=0}^{c_2 \sqrt{n \log n}}Y_i \geq \left(1-\frac{1}{2}\right){c_2}^2 \log n }\leq \exp\left(- {c_2}^2\log(n) / 2\right) \leq \frac{1}{n^{({c_2}^2/2)}}
\]
	
Since each $Y_i$ stochastically dominates over the corresponding $X_i$, it implies 
\[
	\Prob {\sum_{i=0}^{c_2 \sqrt{n \log n}}X_i \geq \left(1-\frac{1}{2}\right){c_2}^2 \log n }\leq \exp\left(- {c_2}^2\log(n) / 2\right) \leq \frac{1}{n^{({c_2}^2/2)}}
\]

Therefore, we can say that with high probability $\sum_{i=0}^{c_2 \sqrt{n \log n}}X_i \leq \left(\frac{1}{2}\right){c_2}^2 \log n $. This means that the number of non-distinct proxies of the contender $y$ is $\leq \left(\frac{1}{2}\right){c_2}^2 \log n$, w.h.p. The statement holds over all contender nodes at the same time by taking a  simple union bound.

Thus, when $t_y=c_3t_{mix}$, each contender $y\in Y$ would have found at least $\frac{c_2}{2}\sqrt{n \log n}$ (which is $<  c_2 \sqrt{n \log n} - ({c_2}^2 \log n)/2$) distinct proxies and therefore satisfying the distinctness property.
\end{proof}

\end{proof}

The time complexity of the algorithm, is determined by the following lemma.

\begin{lemma} [Safety Lemma] \label{lem:safety}
In $O(t_{mix} \log^2 n)$ time, w.h.p. all contender nodes satisfy both the intersection and the distinctness properties. Consequently, for the given algorithm, every node eventually stops, no later than $O(t_{mix}\log^2 n)$ time. %
\end{lemma}

\begin{proof}
Each contender $u$ in parallel, runs several random walks till it satisfies the intersection and distinctness properties. %
$u$ begins with an initial estimate of $1$ and doubles each time till the above condition is not satisfied. This is the standard guess and double strategy and this does not increase the overall complexity by more than a constant factor of the maximum estimate. From Observation \ref{obv:inactive} and Lemma \ref{lem:active}, we see that all contender nodes satisfy both the conditions w.h.p. when $t_u = c_3t_{mix}$ $(c_3 \geq 1)$. Since the algorithm runs an upper-bound of $t_u$, i.e. $T=O(t_u \log^2 n)$ to avoid congestion, the time required to satisfy both the intersection and distinctness property is $O(T)=O(t_{mix}\log^2 n)$. %
\end{proof}

\begin{lemma} [At least one leader] \label{lem:atleast}
After the iteration where the active nodes estimate $t_u = c_3t_{mix}$, where $c_3$ is a constant $\geq 1$, if no node had elected itself as leader in any of the earlier rounds, at least one contender node elects itself as the leader.
\end{lemma}

\begin{proof}

Consider the iteration $i$, where the active nodes estimate $t_u = c_3t_{mix}$ $(c_3 \geq 1)$. We look at the contender node with the highest id, say $v_h$. In iteration $i$, $v_h$ can either be inactive or active depending on whether it has stopped. (If $v_h$ is inactive, we look at the iteration $j$ $(j < i)$ in which $v_h$ became inactive). We show that for either case, if no other node has elected itself as the leader in any of the earlier rounds then $v_h$ becomes leader.

Suppose that $v_h$ becomes inactive in iteration $j$ where $(j < i)$ and no other node has elected itself as the leader in any of the earlier iterations. By Observation \ref{obv:inactive}, $v_h$ satisfies both the intersection and the distinctness properties. %
Alternatively, it could be that $v_h$ is active until iteration $i$, where the active nodes estimate $t_u = c_3t_{mix}$. Also in that case, Lemma~\ref{lem:active} says that $v_h$ satisfies the intersection and the distinctness properties. %
For either case, since $v_h$ has the highest id among the contender nodes, satisfies both the distinctness and the intersection properties and none of the other nodes has elected itself as the leader in any of the earlier rounds (implying that $v_h$ has not received a \textit{winner} message), $v_h$ satisfies all the required conditions and becomes leader.
\end{proof}

\begin{lemma} [At most one leader] \label{lem:atmost}
After the completion of the algorithm, at most one contender node elects itself as the leader.
\end{lemma}

\begin{proof}
We prove the lemma by combining the following two claims:

\begin{claim}
Two different nodes cannot elect themselves as the leader in an iteration of the algorithm.
\end{claim}
\begin{proof}

Suppose two nodes $u$ and $v$ elect themselves as the leader in the same iteration of the algorithm.
We know by the description of the algorithm that any node that becomes the leader would first need to satisfy both the intersection and the distinctness properties. %
Therefore, both contenders $u$ and $v$ would have at least $\tfrac{c_2}{2}\sqrt{n \log n}$ distinct proxies and would be adjacent to $>\frac{3}{4}c_1 \log n$, i.e., more than half of the contenders. Recall that the sets $I_2$ of $u$ and $v$ (denoted by $I_2(u)$ resp.\ $I_2(v)$) contain the ids of their adjacent contenders.
Let $w$ be a contender whose id is in the intersection of $I_2(u)$ and $I_2(v)$. As both $u$ and $v$ are adjacent to more than half of the contenders, there must be at least one such node $w$.

Without loss of generality, assume that the id of $u$ is larger than the id of $w$. Since $w \in ID_2(u)$, some proxy $p_1$ of $u$ must have also been a proxy of $w$ in this iteration. Similarly, since $w \in ID_2(v)$, some proxy $p_2$ of $v$ must have also been a proxy of $w$ in this iteration. Then, by the description of the algorithm $w$ would obtain the ids of both $u$ and $v$ in the set $I_2(w)$, %
 which it then disseminated to all its proxies. %
The proxies $p_1$ and $p_2$ both get this information $I_2(w)$ (of ids of $u$ and $v$) which is then forwarded to $u$ and $v$ respectively as sets $I_3(p_1)$ and $I_3(p_2)$ respectively. This means that $v$ must have known about $u$ while checking the winning condition and hence it knows that its id was not maximal, a contradiction.
\end{proof}

\begin{claim}
If a node elects itself as the leader in some iteration $i$, no other node can elect itself as the leader in any subsequent iteration.
\end{claim}
\begin{proof}
Suppose two nodes $u$ and $v$ elect themselves as the leader and suppose that $u$ does so in iteration $i$ whereas $v$ does so in iteration $j > i$. For this case we show that when iteration $i+1$ begins, more than half of the contender nodes are aware that some node $u$ has become the leader. If any other contender node $v$ satisfies both, the intersection and the distinctness properties, then it must have interacted with at least one of the nodes that is aware of the existence of a leader and thereby also becomes aware of the leader. This means that $v$ must have known about the existence of a leader by receiving a \textit{winner} message (either directly or indirectly), leading to a contradiction. 

If $u$ becomes leader in iteration $i$, then it immediately sends a \textit{winner} message to all its proxies, which is then immediately forwarded it to all the other adjacent contenders (see Algorithm \ref{algo:leader2}). The \textit{winner} message reaches all the adjacent contenders of $u$ before the next iteration begins (as active contenders wait for $2T$ time at the end of random walk phase). As $u$ has to satisfy both the intersection and the distinctness properties to satisfy the winning condition, the number of adjacent contenders of $u$ is $\ge \tfrac{3}{4}c_1 \log n$, which in turn is greater than half of the total number of contenders. Any other contender node that also satisfies the intersection and the distinctness properties has to have at least one intersecting proxy with at least one of the adjacent contender nodes of $u$ (by the pigeon hole principle). Any interaction with adjacent contender nodes of $u$ is accompanied with an additional \textit{winner} message notifying $v$ of the existence of a leader, and thus leading to a contradiction.
\end{proof}
This completes the proof of Lemma~\ref{lem:atmost}.
\end{proof}

Combining Lemma \ref{lem:atleast} and Lemma \ref{lem:atmost}, we obtain the following lemma that determines the correctness of the algorithm.

\begin{lemma} [Unique Leader Lemma]
With high probability and in $O(t_{mix}\log^2 n)$ time, exactly one contender becomes the leader.
\end{lemma}

\begin{lemma} [Message Complexity Lemma] \label{lem:msgcomplexity}
With high probability, the total number of messages sent by the above algorithm is at most $O(\sqrt{n} \log^{7/2} n \cdot t_{mix})$. If larger message size of $O(\log^3n)$ is allowed the total number of messages comes down to $O(\sqrt{n} \log^{3/2} n \cdot t_{mix})$.
\end{lemma}
\begin{proof}
To calculate the message complexity, we look at the various messages that are sent by the algorithm. Considering Algorithm \ref{algo:leader2}, we observe that all the information is sent only along the random walks. The messages that are sent include the random walk tokens, the sets $I_1, I_2$ and $I_3$, the Boolean $d$ and the \textit{winner} messages. 
In each phase (iteration), the maximum number of steps taken by any of these messages is proportional to the estimate of the length of the random walk $t_u$. The maximum possible estimate is $O(t_{mix})$ (c.f. Lemma \ref{lem:safety}) and as this estimate is chosen in a guess-and-double style which only increases the overall complexity to a constant factor of the maximum guess for a successful trial, the overall number of steps taken throughout the algorithm (without accounting for congestion) by any of these messages is $O(t_{mix})$ as well.  

Individually, the Boolean $d$ and the \textit{winner} messages are of $O(1)$ bits and the random walk tokens are of $O(\log n)$ bits. Since the ids of the contenders are of $O(\log n)$ bits and number of contenders is $\le \tfrac{5}{4}c_1 \log n$ (c.f. Lemma \ref{lem:contenders}), the sets  $I_1, I_2$ and $I_3$ can be of size $O(\log^2 n)$ as they can contain the ids of $O(\log n)$ other contenders. This implies that an intermediate node might receive up to $O(\log n)$ many $O(\log ^2 n)$ sized messages. 

\noindent First, consider the case where $O(\log^3 n)$ message sizes are allowed to be sent over an edge.
Each contender node ($O(\log n)$ many) initiates a total of $O(\sqrt{n \log n})$ messages which backtracks after reaching the proxies taking a total of $O(t_{mix})$ steps. Additionally, the \textit{winner} message also takes only $O(t_{mix})$ many steps. As there would be no congestion, the message complexity here would be $O(\log n) \times O (\sqrt{n \log n}) \times O(t_{mix})$ w.h.p. , which equals $O(\sqrt{n} \log^{3/2} n \cdot t_{mix})$.

Now, we consider the standard $\mathcal{CONGEST}$ model where message sizes are restricted to $O(\log n)$. Firstly, during the execution of the random walk, a contender node $u$ does not send $ O (\sqrt{n \log n})$ different tokens for each random walk, but rather sends only one token along with a count of tokens that need to be sent in a particular path. For multiple instances of the variable $d$ originating from different proxies of the same contender, only the summation value is sent (which is $O(\log n)$). Multiple messages coming from either the same or different nodes could possibly lead to congestion. For messages that have the same destination, we send only one distinct copy of id information over a particular edge (i.e. we use a filtering and forwarding technique wherein if an intermediate node has sent the information to a particular destination once it does not send the same information again to that destination). For messages having different destinations, there is a possibility that $O(\log n)$ many messages of $O(\log^2 n)$ size could arrive at a particular intermediate node. Larger sized messages of $O(\log ^2 n)$ bits would have to be broken down into $O(\log n)$ sized messages, i.e. we can assume that each  $O(\log n)$ sized message contains the information of the id of a node and some additional $O(1)$ bits. The maximum delay possible here an at intermediate node is $O(\log^2 n)$. We note that we use the variable $T=O(t_{mix}\log^2 n)$ in the algorithm to deal with this possible delay. Hence, the number of messages sent in the worst case is $O(\log n) \times c_2 \sqrt{n \log n} \times t_{mix} \times O(\log^2 n)$ w.h.p. , which equals $O(\sqrt{n} \log^{7/2} n \cdot t_{mix})$.
\end{proof}

We conclude with the following theorem that combines the results of all the previous lemmas.

\begin{theorem} \label{thm:main}
For any given graph $G$ that has a mixing time of $t_{mix}$, there exists an implicit leader election algorithm that succeeds w.h.p. in $O(t_{mix}\log^2 n)$ time and has a message complexity of $O( \sqrt{n} \log^{7/2} n \cdot t_{mix} )$, assuming that nodes know $n$.%
\end{theorem}

After finding the leader we can use the well known push-pull broadcast \cite{Karp:2000:RRS:795666.796561} to disseminate the id of the leader to all the other nodes to obtain a solution for the explicit variant of leader election. 

\begin{corollary}
For any graph $G$ that has a conductance of $\phi$ and a mixing time of $t_{mix}$, there exists an explicit leader election algorithm that succeeds w.h.p. in $O(t_{mix}\log^2 n)$ time and has a message complexity of $O( \sqrt{n} \log^{7/2} n \cdot t_{mix} + n\frac{\log n}{\phi})$, assuming that nodes know $n$ and there are no failures.
\end{corollary}

\begin{proof}
The corollary follows by appending a simple push-pull broadcast procedure \cite{stacs2011_conductance} at the end of the implicit leader election algorithm. The push-pull broadcast takes $\tfrac{\log n}{\phi}$ time and  $n\tfrac{\log n}{\phi}$ messages. From equation \ref{eq:phimix}, we know that $\Theta(1/\phi) \le t_{mix} \le \Theta(1/{\phi}^2)$, %
it implies $\tfrac{\log n}{\phi} \le O(t_{mix}\log^2 n)$. Therefore the running time of leader election dominates the running time for broadcast.
\end{proof}

\section{Lower bounds}\label{sec:lowerbound}
In this section, we show the lower bounds for implicit leader election by showing that there exists a class of graphs with conductance $\phi$ on which any leader election algorithm that succeeds with probability $1- o(1)$ requires $\Omega\left( {\sqrt{n}}/{(\phi)^{3/4}} \right)$ messages in expectation.
We also obtain some corollaries that lower bound the total number of messages required by other graph problems like broadcast and spanning tree construction.

\begin{theorem} \label{thm:lb}
 Suppose there is a randomized leader election algorithm that succeeds with probability $1 - o(1)$ in $n$-node networks where each node has a unique ID and knows the network size $n$.
  Then, for every $\alpha$, where $\tfrac{1}{n^2} < \alpha < \tfrac{1}{12^2}$, there exists a graph $G$ of $\Theta(n)$ nodes and conductance $\phi = \Theta(\alpha)$ such that
  the algorithm requires $\Omega\left( {\sqrt{n}}/{\phi^{3/4}} \right)$ messages in expectation.
\end{theorem}

We prove the above theorem through a contradiction. Given a particular $n$ and a value of $\alpha$ (within a specified range), we first construct a lower bound graph with $n$ nodes and conductance $\phi = \Theta(\alpha)$. Then,  we assume towards a contradiction that there exists an algorithm that solves implicit leader election on the graph $G$ by sending at most $o({n^{(1-\epsilon)/2}})$ messages in expectation. The key intuition of the proof is to show that given this message budget, different distinct parts of the network are unable to communicate with one another. This lack in communication ensures that either all symmetric parts elect a leader or they do not. If all of the distinct parts do elect a leader it would imply more than one leaders, and if none of them elect a leader it implies zero leaders. Thus, leading to a contradiction.
In the final step of the proof, we leverage the assumed upper bound on the expected message complexity, to show that distinct parts of the network, where nodes might be initiating the exploration of their neighborhoods, are likely to never communicate, and this lack of communication results in having no leaders or multiple leaders with constant probability.

Throughout the proof of Theorem~\ref{thm:lb}, we assume that nodes start \emph{without} unique ids. However, since nodes have knowledge of the network size $n$, it is straightforward to generate unique IDs with high probability. 
Hence we can use the same reduction as \cite{Chatterjee:2018:CLE:3154273.3154308} (Sec.~3, paragraph ``Unique IDs vs Anonymous'')  to remove this assumption and show that our result holds even when nodes are equipped with unique ids.

\subsection{The lower bound graph(s)} \label{sec:lbg}

\para{Graphs $G$ and $\GS$} We start out by describing the construction of the graph $G$ that we use to prove the message complexity lower bound. 
For any given $n$ and $\alpha$ such that $\left(\frac{1}{n^2}\right) < \alpha < \left(\frac{1}{12^2}\right) $, we create the graph $G$ that has a total number of $n$ nodes and a conductance $\phi = \Theta(\alpha)$. In this regard, we also define a parameter $\epsilon = \left(\frac{\log (1/\alpha)}{2\log n}\right)$.

We first construct a super-node graph $\GS$ with $N = \lfloor n^{1-\epsilon} \rfloor$ super-nodes, and later derive the graph $G$ from $\GS$. The graph $\GS$ is created as a random regular graph (as in \cite{bollobas2001random},\cite{Bollobas}) where each super-node has a degree $4$. See Figure \ref{fig:super}. For the purpose of analysis, since it does not change our bounds, we assume that both $n^{1-\epsilon}$ and $n^{\epsilon}$ are integers.

\begin{figure}
	\centering
	\includegraphics[scale=0.3]{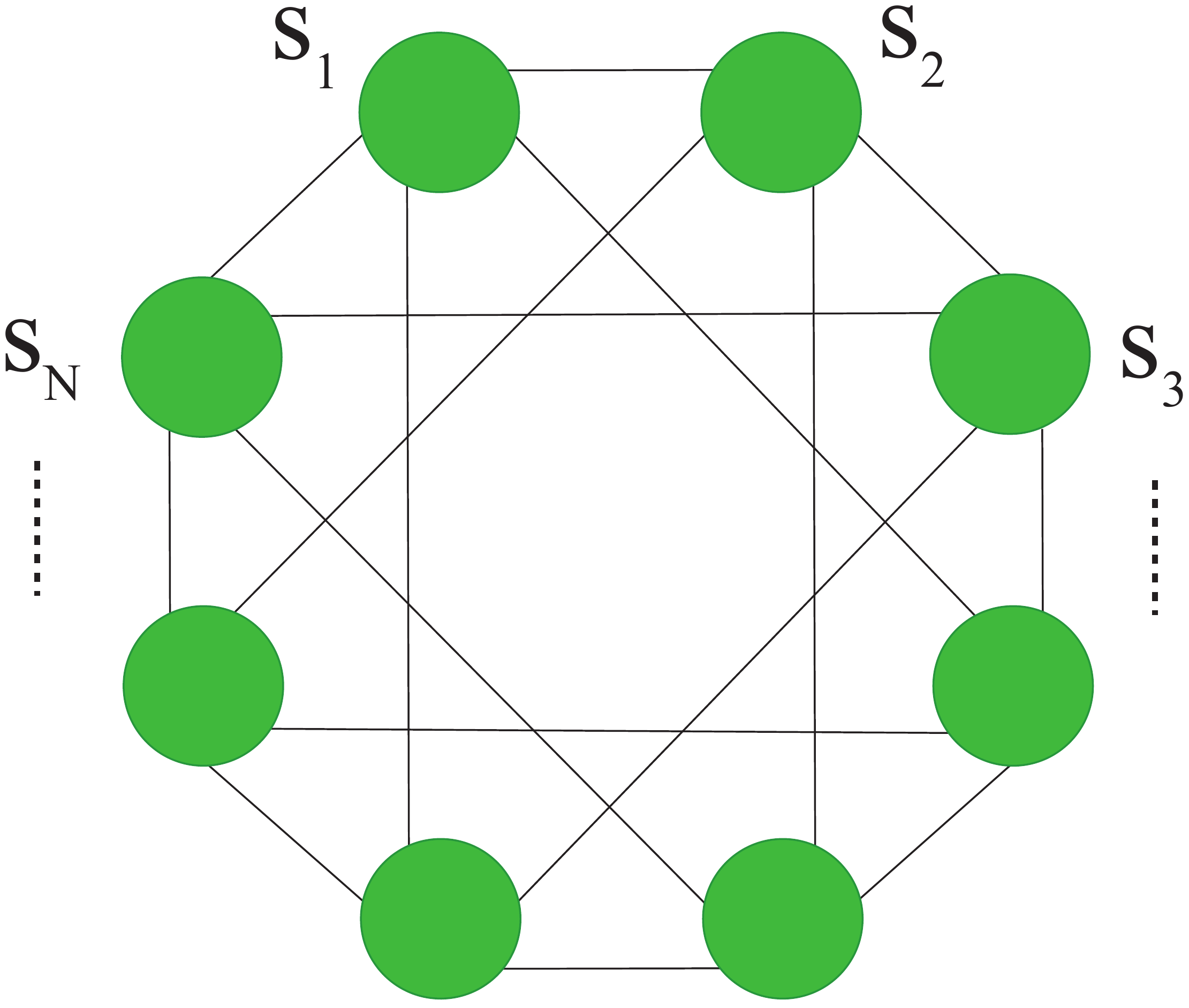}
	\caption{A random 4-regular super-node graph $\GS$ with $N=n^{1-\epsilon}$ super-nodes.}
\label{fig:super}
\end{figure}

Say $V(\GS)= \{s_1,s_2,\dots s_{n^{1-\epsilon}}\}$ and $E(\GS)=\{e_1,e_2,\dots,e_{4n^{1-\epsilon}}\}$ be the vertex set and the edge set of the graph $\GS$. To create the graph $G$ from $\GS$, each super-node $s_i$ is replaced with a clique $C_i$ of $\lceil n^{\epsilon} \rceil$ nodes. For each edge $e_i$ of $\GS$, that exists between two super-nodes say $s_j$ and $s_k$, a corresponding edge $e'_i$ is created in the graph $G$ between a (previously unchosen) node chosen randomly from the clique $C_j$ and a (previously unchosen) node chosen randomly from the clique $C_k$.  %
As each super-node has exactly $4$ edges connected to it, for each clique there would exist $4$ such chosen nodes (called external-edged nodes). An edge between any two nodes belonging to the same clique is called an intra-clique edge, whereas an edge between nodes belonging to different cliques is called an \emph{inter-clique edge}. To maintain uniform node degrees of exactly $ n^{\epsilon} $, two intra-clique edges are removed, one from between any two of the external-edged nodes, and the other from between the remaining two external-edged nodes. See Figure \ref{fig:graph}. %
Thus, in any clique of the graph $G$ there would be two types of nodes, $n^\epsilon -4$ nodes with only intra-clique edges called internal-edged nodes and $4$ nodes  with both intra-clique edges and one inter-clique edge called external-edged nodes.

\begin{figure}
	\centering
	\includegraphics[scale=0.5]{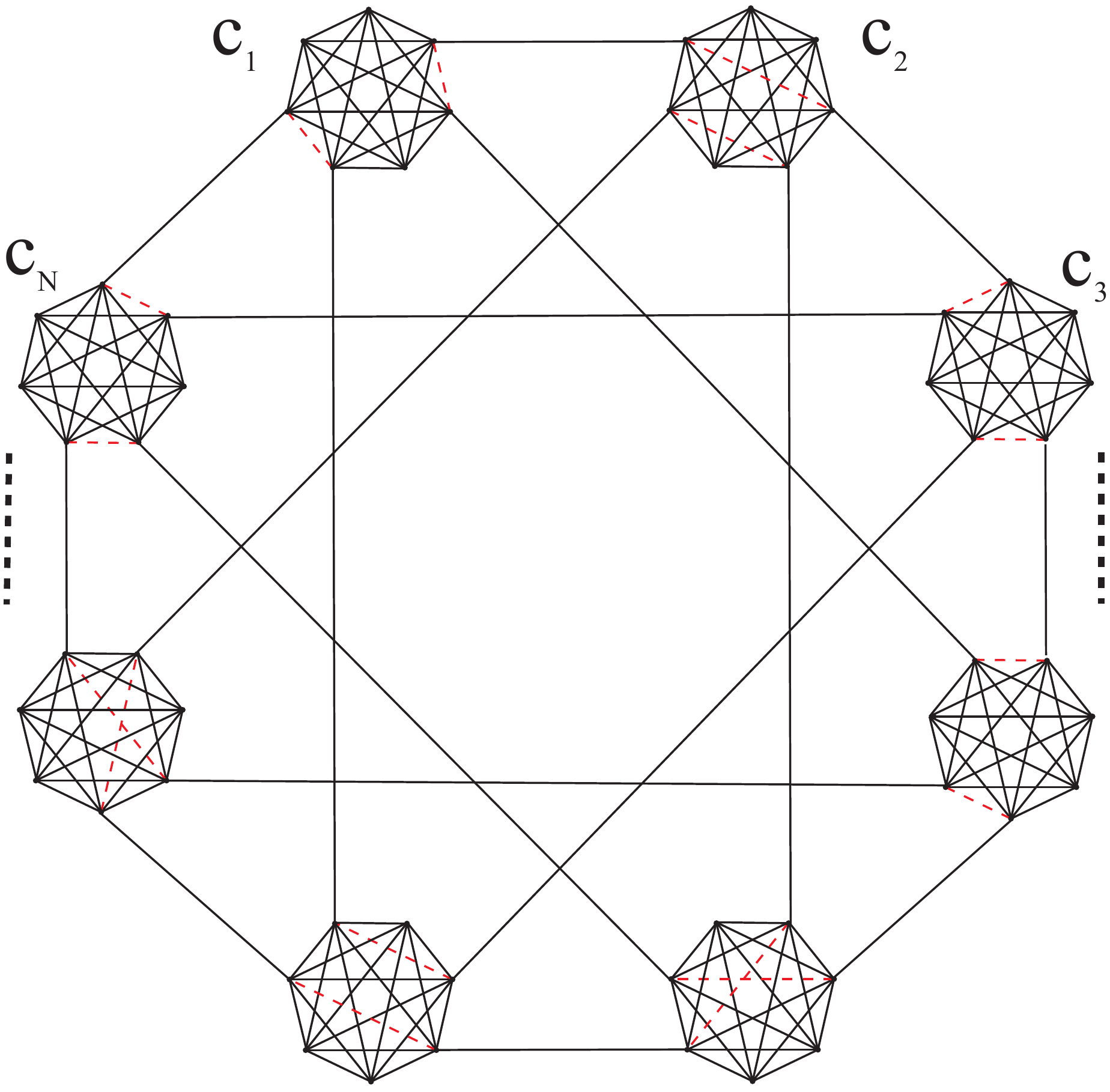}
	\caption{Graph $G$ constructed from $\GS$, where each super-node is replaced by a clique  of size $n^\epsilon$. The red dashed lines indicate the edges of the cliques that are removed to maintain uniform node degree. }
	\vspace{-1em}
\label{fig:graph}
\end{figure}

Based on the construction of $G$, there exists a one to one mapping between the nodes of the super-node graph $\GS$ and the cliques in the graph $G$. 

\para{Graph $\CG$} We define the \emph{clique communication graph $\CG$} as a graph whose vertex set is equivalent to the vertex set of the super-node graph $\GS$, which we simply call \emph{cliques}.
An edge exists in $\CG$ from clique $C_1$ to $C_2$ iff a message is sent on the inter-clique edge from some node in $C_1$ to a node in $C_2$.  
Note that the edge set of $\CG$ can grow over the course of the algorithm.
For the purpose of our analysis, we only keep track of the first message sent on an inter-clique edges and so we treat $\CG$ as a simple  graph. 

\subsection*{High-Level Overview of the Lower Bound Proof} 
We begin by showing that the conductance of the constructed lower bound graph $G$ is $\phi = \Theta(1/n^{2\epsilon})$. Then, we assume towards a contradiction that there exists an algorithm that solves implicit leader election on the graph $G$ by sending at most $M n^{2\epsilon}$ messages in expectation, where $M=o({n^{(1-\epsilon)/2}})$. It implies from the construction of $G$, that $M n^{2\epsilon} = o(\sqrt{n}/\phi^{3/4})$, as $\phi = \Theta(1/n^{2\epsilon})$ and $M=o({n^{(1-\epsilon)/2}})$. Next, on the graph $G$ we show that any algorithm that sends at most $M n^{2\epsilon}$ many messages in expectation, is likely to find at most $O(M)$ inter-clique edges. Then, given the fact that only inter-clique edges can be used for communicating in the clique communication graph $\CG$, we show that the random variables representing the states of the resulting connected components (in $\CG$) are nearly independent of one another. We leverage this ``near independence'' to show that the algorithm is likely to elect either no leader or more than one leader with constant probability (similarly to identically distributed and fully independent indicator random variables), thus resulting in a contradiction. We formalize this overview in the remainder of this section.

\begin{lemma} \label{lem:phi}
  The conductance of the graph $G$ is $\phi = \Theta(\alpha) = \Theta({1}/{n^{2\epsilon}})$ with high probability.
\end{lemma}
\begin{proof}

First, we define an optimal cut of a graph as the cut that determines the minimum cut-conductance of the graph, and hence also determines the conductance of the graph. We prove the lemma by using the following claim.

\begin{claim}
The optimal cut of the graph $G$ does not pass through any of the cliques, i.e., all the edges that are cut by the optimal cut comprises only of inter-clique edges.
\end{claim}

\begin{proof}
Let us assume for the sake of contradiction that a given cut $\mathcal{K}$ of the graph $G$, with cut-conductance $\phi_{\mathcal{K}} < 1/6$, is the optimal cut of $G$. We show that, if $\mathcal{K}$ intersects with (passes through) any clique, the conductance can always be reduced to $\phi_{new}$ such that $\phi_{new} < \phi_{\mathcal{K}}$ by moving a group of nodes from one side of the cut to the other. This will contradict our assumption of $\mathcal{K}$ being the optimal cut and thereby prove the above claim. If $\phi_{\mathcal{K}} \geq 1/6$, we compare this cut to the middle cut of the graph $G$ that cuts the graph into two equal parts and does not pass through any cliques. With a simple calculation it is easily show that this middle cuts' conductance is $< \phi_{\mathcal{K}}$.

We refer to the total volume of the graph (summation of all node degrees) as $V_{total}$. The side of the cut $\mathcal{K}$ that has the (initial) lower volume $(< V_{total}/2)$ is called as the \emph{min} side of the cut and the side (initially) having the larger volume $(> V_{total}/2)$ is called as the \emph{max} side of the cut. If volumes of both sides are equal then we arbitrarily assign one side as \emph{min} and the other as \emph{max}. For the given cut $\mathcal{K}$, we consider $\phi_{\mathcal{K}} = C/V$, where $C$ represents the cut edges (edges with end nodes on either side of the cut $\mathcal{K}$) and $V$ represents the volume of the \emph{min} side of the cut. Note that, since node degrees are uniform, each node contributes $n^{\epsilon}$ towards the volume, and thus the side having the lesser number of nodes has the minimum volume.

We look at the cliques that are cut by $\mathcal{K}$, the side of the clique that has $< n^{\epsilon}/2$ nodes is labeled as the \emph{minority} side and the side of the clique that has $> n^{\epsilon}/2$ nodes is labeled as the \emph{majority} side. If both sides have exactly $ n^{\epsilon}/2$ nodes labeling is done arbitrarily. In this regard, note that whenever we say \emph{min}/\emph{max} side we refer to the side of the cut of the entire graph, whereas when we say \emph{minority}/\emph{majority} side we refer to the cut sides of the particular clique under consideration.

Consider any clique $C_i$ that is divided by the cut $\mathcal{K}$. Let there be $k$ nodes in the minority side of the $C_i$ and $n^{\epsilon}-k$ nodes in the majority side, where $k \leq n^{\epsilon}/2$. To show that we can obtain a lower conductance, we always move $k$ nodes from the \emph{minority} side of the clique to the opposite side of the cut $\mathcal{K}$ (except Case 3, where moving the $k$ nodes leads to the volume of the \emph{min} side of the cut becoming $0$ : in this case $(n^\epsilon -k)$ nodes are moved from the \emph{majority} side of the clique to the other side of the cut). After the nodes are moved, we show that for all cases that the maximum possible value of conductance obtained after moving the nodes ($\phi_{new}$), is less than the conductance prior to the nodes being moved ($\phi_{\mathcal{K}}$), giving us a contradiction. Observe that the denominator is always greater than $0$ as in no case all the nodes are moved out of the eventual \emph{min} side. 

For each case (except Case 3), when $k$ nodes are moved, 
say there are $k_1$ internal-edged nodes and $k_2$ external-edged nodes such that $k_1 + k_2 = k$. Each of the $k_1$ internal-edged nodes in the minority side was previously connected to $n^{\epsilon} -k$ nodes in the majority side (due to the clique edges) and each of the $k_2$ external-edged nodes in the minority side was previously connected to at least $(n^{\epsilon} -k)-1$ nodes in the majority side (case that maximizes $\phi_{new}$). Therefore, comparing $\phi_{\mathcal{K}}$ with $\phi_{new}$, the reduction in the number of cut edges is at least $k_1(n^{\epsilon} -k) + k_2(n^{\epsilon} -k -1) = k(n^{\epsilon}-k) - k_2 $.  The total possible number of external-edged nodes is at most $4$ (as the super-node graph was $4$-regular), therefore in this case $k_2 \leq 4$. Also by our choice of $\epsilon$ (and the range of $\alpha$), we see that $4 < n^{\epsilon}/3 \implies 4 < kn^{\epsilon}/3 $. We consider the following cases for a clique $C_i$ while moving $k$ nodes from the \emph{minority} side to the opposite side of the cut.

\noindent \textbf{Case 1 :} If $k$ nodes move from \emph{min} side of the cut to the \emph{max} side.\\
 Moving $k$ nodes reduces the volume of the \emph{min} side of the cut by $kn^{\epsilon}$.
 \[ \phi_{new} = \frac{C-(k(n^{\epsilon}-k) - k_2)}{V-kn^{\epsilon}}  \leq \frac{C-kn^{\epsilon}+k^2 + 4}{V-kn^{\epsilon}}  < \frac{C-kn^{\epsilon} +kn^{\epsilon}/2 + kn^{\epsilon}/3}{V-kn^{\epsilon}}  = \frac{C-kn^{\epsilon}/6}{V-kn^{\epsilon}} \]
\[ \phi_{\mathcal{K}} - \phi_{new} \geq \frac{C}{V} - \frac{C-kn^{\epsilon}/6}{V-kn^{\epsilon}} = \frac{kn^{\epsilon}(V/6 -C)}{V(V-kn^{\epsilon})} = \frac{Vkn^{\epsilon}(1/6 -\phi_{\mathcal{K}})}{V(V-kn^{\epsilon})}\]
As $\phi_{\mathcal{K}}$ is $< 1/6$, $\phi_{\mathcal{K}} - \phi_{new} > 0$, which implies that $\phi_{new} < \phi_{\mathcal{K}}$.\\
\textbf{Case 2 :} If $k$ nodes move from \emph{max} side of the cut to the \emph{min} side.\\
\textbf{Case 2A :} If \emph{min} side's volume still remains $\leq V_{total}/2$ after the nodes are moved.\\
Moving $k$ nodes into the \emph{min} side of the cut increases the volume by $kn^{\epsilon}$. 
\[ \phi_{new} = \frac{C-k(n^{\epsilon}-k) + k_2}{V+kn^{\epsilon}}  \leq \frac{C-kn^{\epsilon}+k^2 + 4}{V+kn^{\epsilon}}  < \frac{C-kn^{\epsilon} +kn^{\epsilon}/2 + kn^{\epsilon}/3}{V+kn^{\epsilon}}  = \frac{C-kn^{\epsilon}/6}{V+kn^{\epsilon}} \]
The number of cut-edges strictly decrease from previous while the volume increases therefore it is clear that $\phi_{new} < \phi_{\mathcal{K}}$.\\
\textbf{Case 2B :} If \emph{min} side's volume becomes $>V_{total}/2$ due to the $k$ moved nodes.\\
The \emph{max} side now becomes the side with the lower value of volume as a result of $k$ nodes moving to the other side with its volume $= V_{total}-V-kn^{\epsilon}$, where $V_{total}$ is the total volume of the graph.
\[ \phi_{new} = \frac{C-k(n^{\epsilon}-k) + k_2}{ V_{total}-V-kn^{\epsilon}}  \leq \frac{C-kn^{\epsilon}+k^2 + 4}{ V_{total}-V-kn^{\epsilon}}  < \frac{C-kn^{\epsilon} +kn^{\epsilon}/2 + kn^{\epsilon}/3}{ V_{total}-V-kn^{\epsilon}}  = \frac{C-kn^{\epsilon}/6}{ V_{total}-V-kn^{\epsilon}} \]
\[ \phi_{\mathcal{K}} - \phi_{new} \geq  \frac{ \left(CV_{total} - 2CV\right) + \left(Vkn^{\epsilon}/6 -Ckn^{\epsilon}\right) }{V(V_{total}-V-kn^{\epsilon})} = \frac{ C\left(V_{total} - 2V\right) + Vkn^{\epsilon}\left(1/6 -\phi_{\mathcal{K}}\right) }{V(V_{total}-V-kn^{\epsilon})}\]
As $V \leq V_{total}/2$ (because initially $V$ was the volume of the side having the lower volume among the two sides of the cut) and as $\phi_{\mathcal{K}}$ is $< 1/6$, we see that the terms in the numerator,  $C\left(V_{total} - 2V\right)\ge 0$ and $Vkn^{\epsilon}\left(1/6 -\phi_{\mathcal{K}}\right)>0$, i.e. $\phi_{\mathcal{K}} - \phi_{new} > 0$, which implies that $\phi_{new} < \phi_{\mathcal{K}}$.\\
\textbf{Case 3 :} In a special case, when moving $k$ nodes from \emph{minority} side of the clique to the \emph{majority} side results in making the volume of the \emph{min} side of the cut $=0$. (This is the case where the \emph{min} side of the cut contains only the \emph{minority} side of a clique). Instead of moving $k$ nodes from the \emph{minority} side of the clique to the \emph{majority}, $(n^\epsilon - k)$ nodes are moved from the \emph{majority} side to the \emph{minority} side (such that the \emph{min} side of the cut now has exactly one clique). In contrast to all the previous cases, in this case nodes move from the \emph{majority} side of the clique to the \emph{minority} side. It is to be noted that this movement cannot result in increasing the volume of the \emph{min} side of the cut to a value $>V_{total}/2$. Say there are $k_1$ internal-edged and $k_2$ external-edged nodes in the \emph{majority} such that $k_1 + k_2 = (n^\epsilon - k)$. Each of the $k_1$ internal-edged nodes in the \emph{majority} side was previously connected to $k$ nodes in the \emph{minority} side (due to the clique edges) and each of the $k_2$ external-edged nodes in the \emph{majority} side was previously connected to at least $k-1$ nodes in the \emph{minority} side (case that maximizes $\phi_{new}$). Therefore, comparing $\phi_{\mathcal{K}}$ with $\phi_{new}$, the reduction in the number of cut edges is at least $k_1(k) + k_2(k -1) = k(n^{\epsilon}-k) - k_2 $.  The total possible number of external-edged nodes is at most $4$ (as the super-node graph was $4$-regular), therefore in this case $k_2 \leq 4$. Also by our choice of $\epsilon$ (and the range of $\alpha$), we see that $4 < n^{\epsilon}/3 \implies 4 < kn^{\epsilon}/3 $. %
Moving $(n^\epsilon - k)$ nodes into the \emph{min} side of the cut increases the volume by $n^{\epsilon}(n^\epsilon - k)$. 
\[ \phi_{new} = \frac{C-k(n^{\epsilon}-k) + k_2}{V+n^{\epsilon}(n^\epsilon - k)}  \leq \frac{C-kn^{\epsilon}+k^2 + 4}{V+n^{\epsilon}(n^\epsilon - k)}  < \frac{C-kn^{\epsilon} +kn^{\epsilon}/2 + kn^{\epsilon}/3}{V+n^{\epsilon}(n^\epsilon - k)}  = \frac{C-kn^{\epsilon}/6}{V+n^{\epsilon}(n^\epsilon - k)} \]
The number of cut-edges strictly decrease from previous while the volume increases therefore it is clear that $\phi_{new} < \phi_{\mathcal{K}}$.

Since for each case we can obtain a conductance of $\phi_{new} < \phi_{\mathcal{K}}$ if cut $\mathcal{K}$ passes through a clique, we conclude that the optimal cut would not pass through any of the cliques.
\end{proof}

Now, we give a one to one correspondence between the cuts on the super-node graph
$\GS$ and the cuts on $G$ that do not pass through any cliques by considering any
cut $\mathcal{K}_\GS$ of $\GS$ and an identical cut $\mathcal{K}_G$ on $G$ such
that if $\mathcal{K}_\GS$ cuts an edge $e'$ in the graph $\GS$, then the cut
$\mathcal{K}_{G}$ would cut edge $e$ of the graph $G$ that was created in behest
of $e'$ while constructing graph $G$ from $\GS$ (refer to the construction of $G$
described in the beginning of Section \ref{sec:lowerbound}). 

Note that, this also creates a one to one correspondence of their respective
cut-conductances such that $\phi_{\mathcal{K}_G}(G) =
4\phi_{\mathcal{K}_\GS}(\GS)/n^{2\epsilon}$. Clearly, the number of cut edges
across the cuts $\mathcal{K}_\GS$ and $\mathcal{K}_G$ remains same in either case.
Let the volume of the smaller side the cut $\mathcal{K}_\GS$ be $V$. Since each
super-node has a degree $=4$, the total number of super nodes present in the
smaller side of the cut equals $V/4$. As described earlier, while constructing
graph $G$ from $\GS$ each super-node is replaced by a clique of $n^\epsilon$ nodes
with each node having a degree $=n^\epsilon$, wherein the degree is adjusted to
include the $4$ inter-clique edges by removing $2$ intra-clique edges. Therefore,
the volume of the smaller side of the corresponding cut $\mathcal{K}_G$ of the
graph $G$ would be $(V/4)n^{2\epsilon}$. Thus, if the cut conductance determined
by $\mathcal{K}_\GS$ in $\GS$ is $\phi_{\mathcal{K}_\GS}$, it implies that the cut
conductance given by the cut $\mathcal{K}_G$ in $G$ would be
$\phi_{\mathcal{K}_G}=4\phi_{\mathcal{K}_\GS}/n^{2\epsilon}$. 

It immediately follows from the correspondence that if cut $\mathcal{K}_\GS$ is the optimal cut that determines the conductance of the graph $\GS$, then its corresponding identical cut $\mathcal{K}_G$ would be the cut determining the conductance of $G$. 
From \cite{Bollobas}, we know that for a sufficiently large $n$, \textit{almost every} random regular graph with degree $=4$ has a constant conductance which implies that w.h.p. the conductance of $G$, $\phi(G)= \Theta(1/n^{2\epsilon})$. 
\end{proof}

\subsection{Distinct parts remain disjoint}

Recall that we assume $M=o({n^{(1-\epsilon)/2}})$ and, assume towards a contradiction that the algorithm sends at most $M n^{2\epsilon}$ messages in expectation. 
In this section, we show that parts of the network where nodes might be initiating the exploration of their neighborhoods, evolve independently in the sense that they are likely to never communicate.

Let random variable $\msgs$ give the number of messages sent by the algorithm and let random variable $\msgs(C)$ give the number of messages sent by the nodes in clique $C$.

\begin{lemma} \label{lem:unexplored}
Without receiving any messages, if a clique $C$ sends a message\footnote{We slightly abuse notation by saying a clique $C$ sends a message when, in fact, some node in $C$ performs the sending action.} over an inter-clique edge, then it follows that the nodes in $C$ have  sent at least $\Omega(n^{2\epsilon})$ messages in expectation, i.e.  $\expect{\msgs(C)} = \Omega(n^{2\epsilon})$.

\end{lemma}

\begin{proof}
Recall from the construction of the super-node graph $\GS$ that we have assigned the inter-clique ports uniformly at random among all available ports of $C$.
Any clique $C$ has a total of $n^{2\epsilon}$ ports out of which only $4$ ports belong to inter-clique edges. Also, the nodes are unaware of their neighbors' identities, and in particular, the four nodes containing inter-clique port are unaware of this fact. First, we see that if a clique $C$ sends more than $n^{2\epsilon}/2$ messages before sending its first inter-clique message, the lemma is vacuously true. Otherwise, given that no messages were received via an inter-clique edge, it holds that, at any point before sending the first inter-clique edge, there are at least $n^{2\epsilon}/2$ ports among the nodes in $C$ over which no message has been sent yet, and each of them is equally likely to connect to an inter-clique edge. Thus, the probability that a message is sent over an inter-clique edges for the first time (in clique $C$) is at most
$4/(n^{2\epsilon} - \tfrac{n^{2\epsilon}}{2}) = 8/n^{2\epsilon}$.

 Therefore, in expectation the number of messages sent by any clique $C$ before sending its first inter-clique message is at least $ n^{2\epsilon}/8$ which is $\Omega(n^{2\epsilon})$. 
\end{proof}

In the rest of our proof, will analyze the probability that certain subgraphs of the clique communication graph $\CG$ (see section \ref{sec:lbg}) contain a leader node.
We will first state a crucial consequence of Lemma~\ref{lem:unexplored} in the language of clique communication graphs:

\begin{lemma}\label{lem:edges}
With probability $1 - o(1)$, the clique communication graph $\CG$ contains at most $O(M)$ edges.
\end{lemma}
\begin{proof}
Let $B = c M$, for a sufficiently large constant $c$, and suppose towards a contradiction that $\CG$ contains at least $B$ edges with constant probability $\gamma>0$.

For each clique that is in a non-singleton connected component in the clique communication graph $\CG$, we define its \emph{first edge} as the first inter-clique edge over which its nodes have sent (or received) a message to (from) another clique. 	
  Let $F$ be the set of cliques that have first edges. 	
  Since each clique can connect to at most $4$ other cliques in $\CG$, we have $|F| \ge B/4$.	
  We have	
  \begin{align}	
  \expect{\msgs} 	
    &\ge \expect{ \msgs \mid \text{$|F| \ge B/4$ edges}} 	
          \cdot	
          \Prob{\text{$|F| \ge B/4$ edges}} \notag\\	
    &\ge \gamma 	
          \cdot 	
          \expect{ \msgs \mid \text{$ |F|\ge B/4$ edges}} \notag \\	
    &\ge \gamma\sum_{C \in F}	
            \expect{ \msgs(C) \mid \text{$ |F|\ge B/4$ edges}} \label{eq:sumedges}	
  \end{align}	
  Since the number of messages required for discovering an inter-clique edge of $C$ is independent of the event $|F|\ge B/4$, it holds that	
  \[	
  \sum_{C \in F}	
    \expect{ \msgs(C) \mid \text{$ |F|\ge B/4$ edges}} 	
    = 	
  \sum_{C \in F}	
    \expect{ \msgs(C) }.	
  \]	
  Applying Lemma~\ref{lem:unexplored} for each $C \in F$, we obtain from \eqref{eq:sumedges} that	
  \[	
  \expect{\msgs}	=  
    \Omega \left(\gamma n^{2\epsilon} B/4 \right)= c \gamma M n^{2\epsilon} / 4.	
  \]	
By choosing $c$ sufficiently large, we obtain a contradiction to the assumption of sending at most $M n^{2\epsilon}$ messages in expectation.
\end{proof}

\noindent\textbf{Spontaneous Cliques.} Since we consider randomized algorithms, we assume that each node is equipped with a random bit string of infinite length. 
If a clique $C$ does not have any incoming edges in $\CG$  throughout the execution, i.e., it does not receive any messages from nodes in other cliques, then the actions and the state transitions of its nodes depend exclusively on the supplied random bit strings.
In particular, inspecting these random bit strings, we can determine whether nodes in $C$ will send messages across any inter-clique edges of $C$.
This motivates us to call $C$ \emph{spontaneous}, if some node in $C$ eventually sends an outgoing message assuming that no node ever receives an incoming message (as per its initial random string). (Note that it may not actually send an outgoing message because it may receive a message first from some node in another clique.) 
We use the notation $P(C)$ to denote the \emph{connected component} of a clique $C$ in $\CG$ and note that $P(C)$ can grow over time.\vspace{2mm} \\
\noindent\textbf{Disjoint Components.} We define $\disjoint$ to be the event where, at any point in the algorithm's execution, each connected component in $\CG$ contains at most one spontaneous clique, and each non-singleton connected component contains exactly one. This, we show is in fact likely to occur. The next lemma summarizes the main result of this subsection:

\begin{lemma} \label{lem:disjoint}
Event $\disjoint$ occurs with probability $1 - o(1)$.
\end{lemma}
\begin{proof}

From Lemma \ref{lem:edges}, we know that with probability $1 - o(1)$,
the clique communication graph $\CG$ contains at most $c M$ edges, for some fixed constant $c$. 
Also note that, the only way of violating event $\disjoint$, is the merging of two connected components, each that initially had only one spontaneous clique.
Clearly, for each non-singleton connected component, there is at least one
spontaneous node. 
We show that conditioned on the event that $\CG$ has $\le c M$ edges, the probability that a connected component $P$ selects an inter-clique edge to a subgraph $Q$, which can be either another non-singleton
connected component  or a spontaneous clique (that may still be a singleton) is quite low.
Denote this event by $\{ P \rightarrow Q \}$. 
Let $j$ and $k$ be the number of open ports of $P$ and $Q$, respectively.  
Then, it holds that
\[
\prob{[P \rightarrow Q]} \le \frac{k}{N-j} =
O\left(\frac{M}{N-M}\right) = o\left(1/\sqrt{N}\right),
\] 
since $M = o(\sqrt{N})$ by assumption.
This shows that, two
non-singleton connected components do not combine with probability at least $1-o\left(\frac{1}{\sqrt{N}}\right)$.
Considering that there are at most $cM$ possible edges in the clique graph (see Lemma~\ref{lem:edges}), we know that $\disjoint$ occurs with probability at least 
$
  \left( 1-\frac{1}{\sqrt{N}} \right)^{cM} = 1 - o(1).
$
\end{proof}

\subsection{Bounding the dependencies between connected components.}
So far, we have shown that connected components are likely to remain disjoint throughout the execution. 
However, we cannot directly argue that this implies a small probability of electing a leader, since the conditioning on event $\disjoint$ restricts the evolution of a given connected component, as we explain in more detail below.

We view the execution of the algorithm as a sequence of \emph{steps} performed by
cliques, where a step involves either an update to a clique's state (defined below)
or the sending of a message. Note that a \emph{step} here is different from a round as there may be simultaneous actions at cliques
happening in the same round, but we can consider an arbitrary order on such
simultaneous events for analysis.

We define the {state of clique $C$ in $\CG$} as  either (1) \emph{empty}, if
$C$ is not spontaneous, or (2) its state consists of the local states of the nodes that are part of the connected component in $\CG$. 
In this notation, sending a message between two nodes in the same clique corresponds to a local update to the clique's state.

Formally, we use the notation $S(C,t)$ to denote the \emph{state of clique $C$ after $t$ steps} and define $S(t)$ to be the \emph{collective state of all the cliques after $t$ steps}.
By inspecting $S(C,t)$, we can derive whether there is a leader in one of the cliques of the connected component of $C$ in $\CG$.

For the rest of the proof, we assume that all connected components remain disjoint throughout the execution, i.e., event $\disjoint$ occurs (see Lemma~\ref{lem:disjoint}).

Let $\kappa$ be a collection of states after step $t$ for all the cliques and suppose that 
$\kappa$ represents a state in which $\disjoint$ holds; formally, the event $S(t) = \kappa$ has nonzero probability conditioned on $\disjoint$.  
We use the notation $\kappa(C)$ to refer to the state of the clique $C$ in the collection of states $\kappa$.
If the clique nodes eventual states were completely independent, then we would have
$
\prob[S(t) = \kappa] = \prod_{C \in \CG} Pr(S(C,t) = \kappa(C)].
$
Note that, the conditioning on $\disjoint$ can introduce dependencies between the event that some clique transits to a given state and the state of some other cliques and thus we cannot assume that the equality holds.
However, we prove that any possible dependency due to event $\disjoint$, cannot decrease the probability of $S(t)=\kappa$, which is sufficient for our purposes:

\begin{lemma} \label{lem:almostindependent}
Let $\kappa$ be a collection of the clique states after step $t$ that has positive probability of occurring conditioned on $\disjoint$. 
Then, it holds that
\begin{align}
\Prob{S(t) = \kappa} \geq \prod_{C \in \CG} \Prob{S(C,t) = \kappa(C)}. \label{eq:dependent}
\end{align}
\end{lemma}

\begin{proof}
We use induction over the number of steps $t$. 
In the base case, i.e.\ the first step $t=1$, \eqref{eq:dependent} holds with equality as no other steps have been made yet.
Next, we assume the statement holds for step $t-1$ and show that it holds for any step $t\ge 2$:

\begin{align}
\prob{[S(t) = \kappa]} 
 &= \sum_{\text{states } \kappa'}\prob{[S(t-1) = \kappa']} 
     \cdot 
     \prob{[\text{step } t \text{ transitions from }\kappa' \text { to }\kappa]}  \notag\\
 &\geq 
 \sum_{\text{states } \kappa'} 
     \prob{[\text{ step } t \text{ transitions from }\kappa' \text{ to }\kappa ]} 
     \prod_{C' \in \CG} \prob{[S(C',t-1) = \kappa'(C')]},  \label{eq:probprod}
\end{align}
by the inductive hypothesis.
For each possible predecessor state $\kappa'$, the probability of transitioning to $\kappa$ depends on the needed step to move from $\kappa'$ to $\kappa$.  
This, however, depends on a single clique taking a step and changing its state accordingly.  
Let $C$ refer to the clique node that must perform a step to transform $\kappa'$ into $\kappa$ and denote the corresponding event that this happens by $\{\text{$C$ takes step: $\kappa'\rightarrow_t\kappa$}\}$ .  (We ignore states $\kappa'$ from which $\kappa$ is unreachable in one step; obviously their contribution to the probability of $\kappa$ is zero.)
We get
\[
\prob{[\text{ step } t \text{ transitions from }\kappa' \text{ to }\kappa ]} 
= 
\prob[\text{$C$ takes step: $\kappa'\rightarrow_t\kappa$}].
\]
Plugging this into the right-hand side of \eqref{eq:probprod} and factoring out $\prob[S(C, t-1) \!=\! \kappa'(C)]$ from the product, yields
\begin{align}
\prob{[S(t) = \kappa]} 
  &\ge \sum_{\text{$\kappa'$}} 
    \prob[\text{$C$ takes step: $\kappa'\rightarrow_t\kappa$}]
    \cdot 
    \prob[S(C, t-1) \!=\! \kappa'(C)] \!\prod_{C' \ne C}\!\! \prob[S(C',t-1) \!=\! \kappa'(C')]. \notag \\
  &= \!\prod_{C' \ne C}\!\! \prob[S(C',t) \!=\! \kappa(C')]
    \cdot
    \sum_{\text{$\kappa'$}} 
    \prob[\text{$C$ takes step: $\kappa'\rightarrow_t \kappa$}]
    \cdot 
    \prob[S(C, t-1) \!=\! \kappa'(C)], \label{eq:prod-lb}
\end{align}
where the last equality follows because the conditioning on $\disjoint$ tells us that $C$ is the only clique updating its state in step $t$, i.e., $\kappa(C')=\kappa'(C')$ and $S(C',t) = S(C',t-1)$, for all $C' \ne C$. 

To complete the proof, we will show that
\begin{align}
  \prob[\text{$C$ takes step: $\kappa'\rightarrow_t\kappa$}] 
  \ge 
  \prob[S(C, t) = \kappa(C)]. \label{eq:dep-ineq}
\end{align}

In calculating $\prob[\text{$C$ takes step: $\kappa'\rightarrow_t\kappa$}]$, we have to exclude the events that are prohibited by the fact that we have conditioned on $\disjoint$, which implies that this probability depends not just on the state of $C$ after step $t-1$, but also on the other connected components.  
Let $P_{t}(C)$ be the connected component of a spontaneous clique $C$ after step $t$.
Since we condition on event $\disjoint$, it cannot happen that some node in $P_{t-1}(C)$ receives a message from a node in some clique $C'\notin P_{t-1}(C)$, as this would result in a connected component $P_{t}(C)$ having $2$ spontaneous cliques. 
For a similar reason, step $t$ cannot be such that a node in $P_{t-1}(C)$ sends a message to a node in some non-singleton component $P_{t-1}(C')$, where $C' \ne C$.
Thus, we are left with the following two possibilities to show that \eqref{eq:dep-ineq} holds:
\begin{compactenum}
\item Step $t$ concerns only nodes in $P_{t-1}(C)$: 
  In this case, the event corresponding to step $t$ is independent of the state of the cliques not in $P_{t-1}(C)$ and hence \eqref{eq:dep-ineq} holds with equality.
\item Step $t$ consists of some node in $P_{t-1}(C)$ sending a message $m$ to a clique $C' \notin P_{t-1}(C)$, and $C'$ is not part of any non-singleton connected component:
  The left-hand side of \eqref{eq:dep-ineq} assumes that we do not condition on any additional state, and therefore $C'$ can be any of the, say $\ell$, cliques not in $P_{t-1}(C)$. 
  On the other hand, when conditioning on the state of components other than $P_{t-1}(C)$, the number of possible cliques where $m$ can be sent to might be smaller than $\ell$, to avoid hitting a clique that is in some other non-singleton connected component (which would violate $\disjoint$).
  In other words, the number of cliques that $m$ can be sent to cannot increase when we conditioning on additional state on the right-hand side of \eqref{eq:dep-ineq}.
\end{compactenum}
Plugging \eqref{eq:dep-ineq} into \eqref{eq:prod-lb}, we get
\begin{align}
  \prob{[S(t) = \kappa]} 
    &\ge   
      \!\prod_{C \in \CG}\!\! \prob[S(C,t) \!=\! \kappa(C)]
      \cdot
      \sum_{\text{$\kappa'$}} 
      \prob[S(C, t-1) \!=\! \kappa'(C)]. \notag
\end{align}
The lemma follows by using the fact that $\sum_{\text{$\kappa'$}}\prob[S(C, t-1) \!=\! \kappa'(C)] =1$.
\end{proof}

\subsection{Disjoint connected components cannot break the symmetry}

At this point, we have shown that, conditioned on the connected components remaining disjoint, the state of the individual connected components is \emph{almost} independent.  In particular, we have shown that the probability of collectively being in any specific disjoint state is at least as large as the product of the individual probabilities.   Throughout, we are conditioning on the connected components being disjoint.

To complete the proof, we need three further steps.  First, we need to relate the states of the cliques to whether or not a given clique has elected a leader.  Then, we need to relate this almost independent process to a collection of independent random variables that are easier to analyze.  Finally, we show that with constant probability the algorithm elects zero or more than one leaders.  

\paragraph{Leadership.}  We want to analyze the probability of a given set of outcomes in terms of leader election. We define an indicator random variable $Y(C,t)$ such that $Y(C,t) = 1$ if and only if clique $C$ is spontaneous and has a leader in its connected component after step $t$; we simply write $Y(C)$ when $t$ is clear from the context or not important. 
By symmetry, all cliques are identical, and hence are equally likely to be spontaneous and also equally likely to be in a connected component with a leader. 
We define $s$ as the probability of the clique $C$ being spontaneous and $p$ as the probability of the spontaneous clique $C$ having a leader, i.e, $p = \prob[Y(C) = 1 \mid \text{$C$ is spontaneous}]$.  
It follows that $\prob[Y(C)\!=\!1] = sp$.
As noted earlier, observe that with the conditioning on $\disjoint$, 
the $Y$s are \emph{not} necessarily independent. 
For example, the knowledge that $Y(C')=1$, for some clique $C'$, might imply that the connected component of $C'$ has a certain minimum size, which in turn
limits the ways in which the connected component of $C$ can expand in the next step. 

Let $Z$ be a vector of desired outcomes for these indicator random variables, i.e., for each clique $C$ we consider whether $Y(C) = Z_C$.  Let $L(C)$ be the set of states for $C$ compatible with the outcomes $Z_C$, i.e., where component $C$ does or does not elect a leader as specified by $L(C)$.  Let $L$ be the product of all the $L(C)$ subspaces, i.e., $L$ is exactly the set of states compatible with $Z$ for all $C$.  Let $F$ be the state of the algorithm when it stops sending messages.

In the following lemma, we show that the probability of being in one of the states compatible with $Z$ can be decomposed into the probabilities of the individual indicator random variables.  (If the connected components were really independent, it would be exact equality, rather than $\geq$.)

\begin{lemma}
\label{lem:addup}
$\Prob{Y = Z} \geq \prod_{C \in \CG} \Prob{Y(C) = Z_C}$.
\end{lemma}
\begin{proof}
The $\Prob{Y = Z}$ is really the same as $\Prob{F \in L}$, by the way in which we have defined $L$.  We first observe that the probability that $F \in L$ is actually the sum of a collection of disjoint events, i.e., the individual states.  For each, the probability can be decomposed by the near-independence property of Lemma~\ref{lem:almostindependent}.  We then observe that the set $L$ is actually the product of a collection of subspaces, allowing us to rearrange terms and recombine disjoint events. 
\begin{eqnarray*}
\Prob{F \in L} & = & \sum_{f \in L} \Prob{F = f}  \\   %
            & \geq & \sum_{f \in L} \prod_{C \in \CG} \Prob{F(C) = f_C} \hspace{3cm}\text{(by Lemma \ref{lem:almostindependent})} \\%%
            & \geq & \sum_{f_1 \in L(1), f_2 \in L(2), \ldots} \prod_{C \in \CG} \Prob{F(C) = f_C} \\ %
            & \geq & \prod_{C \in \CG} \sum_{f_C \in L(C)} \Prob{F(C) = f_C} \\ %
            & \geq & \prod_{C \in \CG} \Prob{F(C) \in L(C)} \\ %
            & \geq & \prod_{C \in \CG} \Prob{Y(C) = Z_C}
\end{eqnarray*}
\end{proof}

\paragraph{Independent variables.}  Recall that $\Prob{Y(C) = 1} = sp$, where $s$ is the probability that $C$ is spontaneous and $p$ is the probability that a clique elects a leader if it is spontaneous.  (And by symmetry, these are all identical.)  We define a new set of \emph{independent} indicator random variables $X(C)$ where $\Prob{X(C) = 1} = sp$.   

\begin{lemma}
\label{lem:independentX}
For any integer $k$, $\Prob{\sum_C{Y(C)} > k} \geq \Prob{\sum_C{X(C)} > k}$. 
\end{lemma}
\begin{proof}
We show that this follows from Lemma~\ref{lem:addup}, by summing over the collection of outcomes where $\sum Y(C) > k$.
Since there are $N$ cliques in total and we can write $\Prob{\sum_C Y(C) > k} = \sum_{l=k+1}^{N} \Prob{\sum_C Y(C)=l}$, we will first obtain a bound on $\Prob{\sum_C Y(C)=l}$, which is the probability of the event that there are exactly $l$ spontaneous cliques (with leader), for integer $l>0$.

Let $Z$ be the $N$-bit vector of the desired outcomes of the indicator random variables $Y(C)$. By abuse of notation, we can think of $Y$ as a vector of the individual random variables $Y(C)$.
Let $\mathcal{Z}$ be the set of all $N$-bit vectors $Z$ that have a support of exactly size $l$.
It follows that 
\begin{align*}
\Prob{\sum_C Y(C) = l} &= \sum_{Z \in \mathcal{Z}} \Prob{Y=Z} \\
    & \geq \sum_{Z \in \mathcal{Z}}\prod_C \Prob{Y(C) = Z_C} \hspace{1cm} \text{(by Lemma \ref{lem:addup})}\\
											 & = \sum_{Z \in \mathcal{Z}}\prod_C \Prob{X(C) = Z_C} \hspace{1cm} \text{(by def.\ of $X(C)$)}\\
											 & = \sum_{Z \in \mathcal{Z}} \Prob{X=Z} \\
											 & \geq \Prob{\sum_C X(C) = l} 
\end{align*} 

Plugging this bound into $\Prob{\sum_C Y(C) > k}$, we obtain
\begin{align*}
\Prob{\sum_C Y(C) > k} = \sum_{l=k+1}^{N} \Prob{\sum_C Y(C)=l} \geq \sum_{l=k+1}^{N} \Prob{\sum_C X(C)=l} = \Prob{\sum_C X(C) > k}.
\end{align*}
This completes the proof of the lemma.
\end{proof}

\paragraph{Zero leaders.}  We now analyse the probability that there are zero leaders, and use that to show that $sp \geq 1/n^{1-\epsilon}$.

\begin{lemma}
\label{lem:zero}
$sp \geq 1/n^{1-\epsilon}$
\end{lemma}
\begin{proof}
We first show that $\Prob{\sum_C Y(C) = 0} \geq (1 - sp)^{n^{1-\epsilon}}$.  We then use this to conclude that $sp \geq 1/n$.
Let us consider $W$ as a subset of cliques of $\CG$. When saying \emph{$W$ is spontaneous}, we mean that all cliques in $W$ are spontaneous.
\begin{eqnarray}
\Prob{\sum_C Y(C) = 0}	
	& \geq & \sum_{W \subseteq \CG} \Prob{\textrm{W are spontaneous}}\Prob{\textrm{no leaders in W}} \nonumber \\
	& \geq & \sum_{W \subseteq \CG} \Prob{\textrm{W are spontaneous}}(1-p)^{|W|} \nonumber \\
	& \geq & \sum_{W \subseteq \CG} \Prob{\textrm{W are spontaneous}} \prod_{C \in W} (1-p) \nonumber\\
	& \geq & \sum_{w_1 \in \{0,1\}, w_2 \in \{0,1\}, \ldots} \Prob{\textrm{$C$ spontaneous iff $w_C = 1$}} \prod_{C \in W} (1-p) \label{eq:8}\\
	& \geq & \sum_{w_1 \in \{0,1\}, w_2 \in \{0,1\}, \ldots}\quad \prod_{C : w_C = 1} s(1-p) \prod_{C : w_C=0} (1-s) \label{eq:9}\\
	& \geq & \prod_{C \in \CG} (s(1-p) + (1-s)) \nonumber\\ 
	& \ge & \prod_{C \in \CG} (1 - sp) \nonumber\\
  	& \ge & (1-sp)^{n^{1-\epsilon}} \nonumber
\end{eqnarray}
Equation \ref{eq:8} follows as we can introduce an indicator variable $w_C$ for each clique $C$ in $\CG$, where $w_C = 1$ iff the clique $C$ is spontaneous. Rearranging the equation for the two different possible values of $w_C$ and observing that $w_C=1$ for all $C \in W$, we obtain Equation~\ref{eq:9}.

Specifically, we know that, for any algorithm to succeed the probability of zero leaders has to be less than a constant. Here, we have shown that any algorithm that sends $\leq Mn^{2\epsilon}$ messages in expectation has at most $O(M)$ edges in the clique communication graph with probability $1-o(1)$, and the connected components formed in the clique communication graph are also disjoint with probability $1-o(1)$.  Conditioned on those events, we have just shown that with probability $(1-sp)^{n^{1-\epsilon}}$ there are no leaders and the algorithm fails.  Thus, for this to be smaller than some constant, we conclude that $sp > 1/n^{1-\epsilon}$. 
\end{proof}
This proves something that intuitively makes sense: in order to ensure at least one leader, if the probability of a clique electing a leader is $sp$, and if we have $n^{1-\epsilon}$ cliques, then the probability $sp \geq 1/n^{1-\epsilon}$ to ensure at least one leader.  (It required just a bit more care because we did not have complete independence.)

\paragraph{More than one leader.}  We now analyze the probability that there is more than one leader, showing that this occurs with constant probability.  This concludes our proof, as it indicates that the algorithm fails to elect exactly one leader with constant probability.  

\begin{lemma} \label{lem:morethanone}
$\Prob{\sum_C{Y(C)} > 1} > \Omega(1)$
\end{lemma}
\begin{proof}
We assume from Lemma~\ref{lem:zero} that $sp > 1/n^{1-\epsilon}$.  We know from Lemma~\ref{lem:independentX} that $\Prob{\sum_C{Y(C)} > 1} \geq \Prob{\sum_C{X(C)} > 1}$.  So we are going to analyze the probability that $\sum_C{X(C)} > 1$.  And this is a straightforward analysis of independent random variables.  

The probability that all the $X(C)$ are 0 is at most:
$$
(1 - sp)^{n^{1-\epsilon}} \leq (1 - 1/n^{1-\epsilon})^{n^{1-\epsilon}} \leq 1/e\ .
$$
(This relies on the fact that $sp \geq 1/n^{1-\epsilon}$.)

We can also analyze the probability that there is exactly one $C$ where $X(C) = 1$.  Specifically, this occurs with probability:
$$
n^{1-\epsilon}sp(1 - sp)^{n^{1-\epsilon}-1}\ .
$$ 
This is maximized when $sp = 1/n^{1-\epsilon}$, so we conclude that
$$
\Prob{\sum_C{X(C)} = 1} \leq n^{1-\epsilon}(1/n^{1-\epsilon})(1 - 1/n^{1-\epsilon})^{n^{1-\epsilon}-1} \leq  1/e + o(1).
$$
Finally, then, we conclude that $\Prob{\sum_C X(C)  > 1} \geq 1 - 2/e - o(1)$.  That is, with at least constant probability there is more than one $X(C) = 1$, and hence   
$\Prob{\sum_C(Y(C)) > 1} \ge \Omega(1)$.
With constant probability, the algorithm elects more than one leader.
\end{proof}

We conclude that if a given algorithm sends at most $Mn^{2\epsilon}$ messages in expectation, then, with constant probability, it either elects zero leaders or more than one leader, thus resulting in a contradiction.
This completes the proof for lower bounding the number of messages required for implicit leader election.
However, for the purpose of readability we restate the theorem and the proof outline.

\begin{reptheorem} {thm:lb}
 Suppose there is a randomized leader election algorithm that succeeds with probability $1 - o(1)$ in $n$-node networks where each node has a unique ID and knows the network size $n$.
  Then, for every $\alpha$, where $\tfrac{1}{n^2} < \alpha < \tfrac{1}{12^2}$, there exists a graph $G$ of $\Theta(n)$ nodes and conductance $\phi = \Theta(\alpha)$ such that
  the algorithm requires $\Omega\left( {\sqrt{n}}/{\phi^{3/4}} \right)$ messages in expectation.
\end{reptheorem}

\begin{proofoutline}
First, in section \ref{sec:lbg}, we construct a lower bound graph $G$ for which the above theorem would hold. By lemma \ref{lem:phi}, we know that graph $G$ has conductance $\phi = \Theta(1/n^{2\epsilon})$. 
We then assume towards a contradiction that there exists an algorithm that solves implicit leader election by sending at most $Mn^{2\epsilon}  = o(\sqrt{n}/\phi^{3/4})$ messages in expectation, where $M=o({n^{(1-\epsilon)/2}})$. %
Next, %
lemma \ref{lem:edges} shows that any algorithm sending at most $O(Mn^{2\epsilon})$ messages in expectation would have at most $O(M)$ edges in the clique communication graph $\CG$ with probability $1-o(1)$. Given this, we know from lemma \ref{lem:disjoint} that event $\disjoint$ occurs with probability $1-o(1)$, i.e. connected components (distinct parts of the network) do not merge (communicate). Thereafter, lemma \ref{lem:almostindependent} and \ref{lem:addup} show that these distinct parts are nearly independent i.e. the random variables representing the states of the resulting connected components (in $\CG$) are nearly independent of one another (behaves similarly to identically distributed and fully independent indicator random variables). Lemma \ref{lem:zero} and \ref{lem:morethanone} leverage this near independence to show that the algorithm is likely to elect either no leader or more than one leader with constant probability. This results in a contradiction and completes our proof. 
\end{proofoutline}

We also obtain the following corollaries that lower bounds the number of messages required for any algorithm that solves broadcast or constructs a spanning tree. %
On the constructed lower bound graph $G$, as opposed to implicit leader election, an algorithm for either broadcast or spanning tree construction would need to discover all $N = n^{1-\epsilon}$ cliques instead of just $\sqrt{N}$ cliques. In the following corollary, we lower bound the number of messages required for any broadcast algorithm.

\begin{corollary}
  
Suppose that there is a randomized broadcast algorithm that succeeds with probability $1 - o(1)$.
Then, for every $\alpha$, where $\left({1}/{n^2}\right) < \alpha < \left({1}/{12^2}\right)$, there exists a graph $G$ of $\Theta(n)$ nodes and conductance $\phi = \Theta(\alpha)$ such that the algorithm requires $\Omega\left( {n}/{\sqrt{\phi}} \right)$ messages in expectation.  
\end{corollary}

\begin{proof}
In the constructed lower bound graph $G$ (described in Section \ref{sec:lbg}), observe that any broadcast algorithm would need to find all the $N = n^{1-\epsilon}$ cliques. As shown in Lemma \ref{lem:unexplored}, we see that discovering a yet undiscovered clique requires $ \Omega(n^{2\epsilon})$ messages. Consequently, the total number of messages required to find all the cliques is $\Omega({n^{1-\epsilon}}\cdot n^{2\epsilon})$ in expectation. From Lemma \ref{lem:phi}, we know that the conductance of the graph $G$ is $\phi = 1/n^{2\epsilon}$. Therefore, the algorithm would require $\Omega(n\cdot n^\epsilon) = \Omega (n/\sqrt{\phi})$ messages in expectation.
\end{proof}

We repeat the same argument to give a message complexity lower bound for spanning tree construction.

\begin{corollary}
 Suppose that there is a randomized spanning tree construction algorithm that succeeds with probability $1 - o(1)$.
Then, for every $\alpha$, where $\left({1}/{n^2}\right) < \alpha < \left({1}/{12^2}\right)$, there exists a graph $G$ of $\Theta(n)$ nodes and conductance $\phi = \Theta(\alpha)$ such that the algorithm requires $\Omega\left( {n}/{\sqrt{\phi}} \right)$ messages in expectation.  
\end{corollary}
\begin{proof}
In the constructed lower bound graph $G$ (described in Section \ref{sec:lbg}), observe that any spanning tree construction algorithm would need to find at least $N-1 = O(n^{1-\epsilon})$ cliques. As shown in Lemma \ref{lem:unexplored}, we see that discovering a yet undiscovered clique requires $ \Omega(n^{2\epsilon})$ messages. Consequently, the total number of messages required to find all the cliques is $\Omega({n^{1-\epsilon}}\cdot n^{2\epsilon})$ in expectation. From Lemma \ref{lem:phi}, we know that the conductance of the graph $G$ is $\phi = 1/n^{2\epsilon}$. Therefore, the algorithm would require $\Omega(n\cdot n^\epsilon) = \Omega (n/\sqrt{\phi})$ messages in expectation.
\end{proof}

\section{The critical knowledge of the network size}
In this section, we show that the knowledge of the network size $n$ is critical for our algorithm to succeed by giving a message complexity lower bound of $\Omega(m)$ for all graphs if $n$ is not known.

In \cite{Kutten:2015:CUL:2742144.2699440}, Kutten et al. show a message complexity lower bound of $\Omega(m)$ in expectation for implicit leader election in general graphs (where $m$ refers to the number of edges in the network graph) even when the number of nodes in the network $n$ and the diameter of the network $D$ are known to all the nodes. Here, we show that this lower bound applies only to graphs that are not well connected or where nodes are not aware of the value of $n$. This lower bound fails for the case of well-connected graphs for the case where $n$ is known (as shown by our algorithm). However, we would like to point out that the knowledge of $n$ is critical for our algorithm to succeed. %

Consider any $2$-connected graph $G_0$ of $n$ nodes, where nodes do not know the value of $n$ and a range $Z = [1, n^4]$ of ID's. $G_0$ can have many instantiations, depending upon the node ID assignment and the port number mapping. An ID assignment is a function $\varphi$ : $V (G_0) \mapsto Z$. A port mapping for node $v$ is a mapping $P_v : [1, deg_v] \mapsto \Gamma(v)$ (namely, $v$'s neighbors). A port mapping for the graph $G_0$ is $P = \langle P_{v_1} , \dots , P_{v_n}\rangle$. Every choice of $\varphi$ and $P$ yields a concrete graph $G_{\varphi,P}$.

\begin{theorem} \label{thm:needn}
Let $\mathcal{R}$ be any implicit leader election algorithm that succeeds with probability at least $1 -\beta$, for some constant $\beta \leq 3/56$. If $n$ is not known to the nodes, for any $2$-connected graph $G_0$ of $n$ nodes and $m$ edges, there exists an id assignment and a port mapping, for which the expected number of messages used by $\mathcal{R}$ on $G$ is $\Omega(m)$.
\end{theorem}

\begin{proof}
To show the lower bound, we rely on the construction of a graph family referred to as dumbbell graphs and on a solution of an intermediate problem called bridge crossing on this graph family. In this regard, we reuse some of the work done in \cite{Kutten:2015:CUL:2742144.2699440} to show that solving bridge crossing on this dumbbell graph family requires $\Omega(m)$ messages in expectation. For completeness, we rewrite some of the definitions and lemmas used in \cite{Kutten:2015:CUL:2742144.2699440}.

Given any $2$-connected graph $G_0$, lets $\mathcal{G}$ be the collection of concrete graphs $G_{\varphi,P}$ obtained from $G_0$ by fixing the node id assignment and the port number mapping. The set of id's of this graph is denoted by $ID(G_{\varphi,P}) = \{\varphi(v) | v \in V(G_0)\}$.
An ``open graph" $G[e]$ is obtained from a graph $G \in \mathcal{G}$ by erasing an edge $e$ of $G_0$ and leaving the two ports that were attached to it empty. 
Let $G^{open}$ be the collection of open graphs obtained from $G_0$.

For two open graphs $G'[e']$ and $G''[e'']$ with disjoint sets of id's, $ID(G'[e']) \cap ID(G''[e'']) = \varnothing$, let $Dumbbell(G'[e'], G''[e''])$ be the graph obtained by taking one copy of each of these graphs, and connecting their open ports. Hence, a dumbbell graph is composed of two open graphs plus two connecting edges, referred to as bridges. Moreover, we say that $G'[e']$ participates on the left and $G''[e'']$ participates on the right in $Dumbbell(G'[e'], G''[e''])$. Strictly speaking, there could be two such graphs, but let us consider only one of them.

For concreteness, if $e' = (v',w')$ and $e'' = (v'', w'')$ where $ID(v') < ID(w')$ and $ID(v'') < ID(w'')$, then the graph $Dumbbell(G'[e'], G''[e''])$ contains the bridge edges $(v',v'')$ and $(w',w'')$. We create a collection $\mathcal{I}$ of inputs for our problem consisting of all the dumbbell graphs. 
\[\mathcal{I} = {Dumbbell(G'[e'], G'' [e'']) \bigm| G' [e'], G'' [e''] \in G^{open} ,
ID(G' [e']) \cap ID(G''[e'']) = \varnothing}.\]
Partition the collection of inputs $\mathcal{I}$ into classes as follows: for every two graphs $G' , G'' \in G$, define the class $C(G',G'') = {Dumbbell(G'[e'], G'' [e'']) \bigm| e',e'' \in E(G')}$, consisting of the $m^2$ dumbbell graphs constructed from $G'$ and $G''$. Finally, create a uniform distribution $\Psi$ on $\mathcal{I}$.

Similar to in \cite{Kutten:2015:CUL:2742144.2699440}, we define an intermediate problem on the input collection $\mathcal{I}$, called \textit{bridge crossing} (BC). An algorithm for this problem is required to send a message on one of the two bridge edges connecting the two open graphs (from either direction). More precisely, any algorithm solving BC is allowed to start simultaneously at all nodes, and succeeds if during its execution, a message has crossed one of the two connecting bridge edges. (Note that in our model, the nodes are unaware of their neighbors' identities, and in particular, the four nodes incident to the two bridge edges are unaware
of this fact.)

We now give a high level overview of the main ideas of the proof. %
For any given $2$-connected graph $G_0$, we show that there exists a graph $Dumbbell(G_l[e_l], G_r [e_r])$ in the collection $\mathcal{I}$ corresponding to $G_0$, for which any algorithm that solves bridge crossing requires $\Omega(m)$ messages in expectation. Here $G_l$ and $G_r$ are graphs obtained from $G_0$ by some id assignment and port mapping; $G_l[e_l], G_r [e_r]$ are their corresponding open graphs obtained by removing edges $e_l$ and $e_r$ respectively. The existence of the dumbell graph follows from the following lemma (from \cite{Kutten:2015:CUL:2742144.2699440}) and  Yao's minmax principle (c.f. Prop. 2.6 in \cite{Motwani:1995:RA:211390}). Thereafter, we give an indistinguishability argument in which we show that if $n$ is not known, no algorithm can distinguish between graphs $G_{l}$ and $Dumbbell(G_l[e_l], G_r [e_r])$ (or $G_{r}$ and $Dumbbell(G_l[e_l], G_r [e_r])$) with sufficiently large probability by sending only $o(m)$ messages. Based on this indistinguishability argument and the fact that BC requires $\Omega(m)$ messages, we show that one side of the dumbbell graph (either $G_l$ or $G_r$) would need to send at least $m/2$ messages to solve leader election.  %

\begin{lemma} \text{(Lemma $3.6$ of \cite{Kutten:2015:CUL:2742144.2699440})}
Every deterministic algorithm $\mathcal{D}$ that achieves BC on at least $1/4$ of the dumbbell graphs in the collection $\mathcal{I}$ has expected message complexity $\Omega(m)$ on $\Psi$.
\end{lemma}

\noindent Combining the above lemma with Yao's minmax principle, we obtain the following lemma that describes the message complexity of any algorithm for BC (both deterministic and randomized) that succeeds with sufficiently high probability ($>5/8$) on the worst case graph of $\mathcal{I}$.

\begin{lemma} \label{lem:comb}
Any algorithm $\mathcal{A}$ that solves BC with probability $>5/8$ on the worst-case graph of $\mathcal{I}$ (say $Dumbbell(G_l[e_l], G_r [e_r])$) has expected message complexity of at least $\Omega(m)$.
\end{lemma}

Consider a universal leader election algorithm $\mathcal{R}$ that succeeds on any given graph $G$ with probability at least $1-\beta$, where $\beta \leq 3/56$. We imagine running $\mathcal{R}$ in parallel on all three graphs ($Dumbbell(G_l[e_l], G_r [e_r]), G_{l}$ and $G_{r}$) using the same random bits. Let $\mathcal{X}$ be the time-point where algorithm $\mathcal{R}$ achieves BC on the graph $Dumbbell(G_l[e_l], G_r [e_r])$ (if $\mathcal{R}$ does not achieve BC, we consider $\mathcal{X} = \infty$). Also, let the expected number of messages sent by algorithm $\mathcal{R}$ on any graph $G$ until time-point $t$ be represented by $msg_{t}(G)$, and so $msg_{\mathcal{X}}(Dumbbell(G_l[e_l], G_r [e_r]))$ is the expected number of messages sent by $\mathcal{R}$ up to the time-point $\mathcal{X}$ on the graph $Dumbbell(G_l[e_l], G_r [e_r])$. 

Observe that until $\mathcal{X}$, the nodes in $G_{l}[e_l]$ (the left side of $Dumbbell(G_l[e_l], G_r [e_r])$) are not aware of the existence of $G_{r}[e_r]$ (as no message has traveled across the bridge edges and $n$ is also not known). As $G_{l}$ has the exact same ids as  $G_{l}[e_l]$ and $\mathcal{R}$ uses the same random bits, until the point $\mathcal{X}$, nodes behave identically in both cases. That is, after $t$ steps (where $t<\mathcal{X}$), if node $x$ in $G_l[e_l]$ is in state $\sigma$, then node $x$ of $G_{l}$ would also be in state $\sigma$. The same argument follows for $G_r [e_r]$ and $G_{r}$. This implies that until BC, the state of any node in $G_l$ (resp. $G_r$) is identical to its corresponding node in $Dumbbell(G_l[e_l], G_r [e_r])$ and as such, the behavior of the nodes would be identical.

\begin{observation} \label{obv:sim}
If $n$ is not known, any algorithm $\mathcal{A}$ (using the same random bits) cannot differentiate if it is running on $Dumbbell(G_l[e_l], G_r [e_r])$ or on $G_{l}$ (resp. $G_{r}$) until the time-point $\mathcal{X}$, when it achieves BC on $Dumbbell(G_l[e_l], G_r [e_r])$. Therefore, until the point $\mathcal{X}$ the behavior of any node in $G_l$ (resp. $G_r$), would be identical to that of the corresponding node in $Dumbbell(G_l[e_l], G_r [e_r])$.
\end{observation}

Let $\mathsf{succ}$ be the event that algorithm $\mathcal{R}$ successfully elects a unique leader in all the three graphs $G_l$, $G_r$ and $Dumbbell(G_l[e_l], G_r [e_r])$ within finite time. %
Conditioned on $\mathsf{succ}$, let $\mathcal{Y}_1$ and $\mathcal{Y}_2$ be the time-points where $\mathcal{R}$ solves leader election on $G_{l}$ and $G_{r}$ respectively ($\mathcal{Y}_1 \neq \infty$ and $\mathcal{Y}_1 \neq \infty$). That is, after $\mathcal{Y}_1$ (resp.  $\mathcal{Y}_2$), exactly one node elects itself as the leader and no more messages are sent on $G_l$ (resp. $G_r$).  %
We also define $\mathcal{Y} = \max(\mathcal{Y}_1, \mathcal{Y}_2)$. %
Using the fact that bridge crossing requires $\Omega(m)$ messages on $Dumbbell(G_l[e_l], G_r [e_r])$, we will show that leader election takes at least $\Omega(m/2)$ messages either for graph $G_{l}$ or for graph $G_{r}$. In this regard, we consider two different possibilities, each of which is conditioned on the event $\mathsf{succ}$. 
\\
\textbf{Case 1 :} $\mathcal{Y} < \mathcal{X}$. %
We show that this case is not possible by showing a contradiction. As $\mathcal{Y} < \mathcal{X}$ %
, it implies that $\mathcal{R}$ solves LE on both $G_{l}$ and $G_{r}$ before solving BC on the graph $Dumbbell(G_l[e_l], G_r [e_r])$. From Observation \ref{obv:sim}, we know that $\mathcal{R}$ cannot differentiate if it is running on $Dumbbell(G_l[e_l], G_r [e_r])$ or on $G_{l}$ (resp. $G_{r}$) until the point $\mathcal{X}$. Therefore algorithm $\mathcal{R}$ on $Dumbbell(G_l[e_l], G_r [e_r])$ would behave in an identical fashion with that in $G_{l}$ (resp. $G_{r}$) and would elect two leaders (one from $G_{l}$ and other from $G_{r}$). Thus contradicting our assumption of event $\mathsf{succ}$. This also tacitly implies that if event $\mathsf{succ}$ happens then $\mathcal{R}$ also solves BC (as $\mathcal{X}$ cannot be $\infty$).
\\
\textbf{Case 2 :} $\mathcal{Y} \geq \mathcal{X}$. %
This case implies that either both $\mathcal{Y}_1 \geq \mathcal{X}$ and $\mathcal{Y}_2 \geq \mathcal{X}$ or only either one of them is $\geq \mathcal{X}$. %
Using Observation \ref{obv:sim}, we can say that the total number of messages sent upto the point $\mathcal{X}$ by running $\mathcal{R}$ on $G_{l}$ and $G_{r}$ is exactly equal to the number of messages sent by running $\mathcal{R}$ on $Dumbbell(G_l[e_l], G_r [e_r])$ until time $\mathcal{X}$. There also might be some additional messages sent on  $G_l$ and/or $G_r$ as $\mathcal{Y} \geq \mathcal{X}$. Therefore, the total number of messages sent by $\mathcal{R}$ on $G_{l}$ and $G_{r}$ would be at least $msg_{\mathcal{X}}(Dumbbell(G_l[e_l], G_r [e_r]))$. By Lemma \ref{lem:comb}, we know that  $msg_{\mathcal{X}}(Dumbbell(G_l[e_l], G_r [e_r]))\geq \Omega(m)$. Thus, we see that  $msg_{\mathcal{Y}_1}(G_{l}) + msg_{\mathcal{Y}_2}(G_{r}) \geq  msg_{\mathcal{X}}(Dumbbell(G_l[e_l], G_r [e_r])) \geq \Omega(m)$. 

That is, conditioned on the event $\mathsf{succ}$, either $msg_{\mathcal{Y}_1}(G_{l}) \geq \Omega(m/2)$ or $msg_{\mathcal{Y}_2}(G_{r}) \geq \Omega(m/2)$. Without loss of generality, assume that $msg_{\mathcal{Y}_1}(G_{l}) \geq msg_{\mathcal{Y}_2}(G_{r})$. Also, since $\mathcal{R}$ is a universal leader election algorithm that succeeds on any given graph with probability at least $1-\beta$, where $\beta \leq 3/56$, then the probability that event $\mathsf{succ}$ happens would be $\geq (1-\beta)^3 \geq 4/5$. This implies from above $\expect{\text{Messages sent on } G_L \mid \mathsf{succ}} \ge \Omega(m/2)$. To calculate the value of $\expect{\text{Messages sent on } G_L}$, we use $\expect{\text{Messages sent on } G_L} \geq \expect{\text{Messages sent on } G_L \mid \mathsf{succ}}\cdot \prob(\mathsf{succ}) \geq \Omega(m/2)\cdot 4/5 = \Omega(m)$. We consider the $msg_{\mathcal{Y}_1}(G_{l}) = \Omega(m)$ as the worst case message complexity and the corresponding graph as the worst case graph. The existence of this worst-case graph proves the theorem.
\end{proof}

\section{Conclusion}
In this paper we show that implicit leader election can be achieved in sub-linear message complexity for sufficiently well-connected graphs. This shows that the major communication cost for the explicit variant of the leader election comes from broadcasting the leader information to all the nodes rather than the process of electing a leader.

%The model used here, where nodes only have local knowledge is referred to as the clean network model in \cite{peleg} and is also otherwise referred to as the $KT_0$ model \cite{DBLP:journals/jacm/AwerbuchGPV90, Pandurangan:2017:TMD:3055399.3055449}. We believe that our bounds can be extended to the $KT_1$ model, where nodes are aware of the identities of their neighbors.

Furthermore, we observe that that there exists a possible gap of $O(1/\phi^{5/4})$ between the upper and the lower bounds shown here. It remains an interesting open problem to see if this gap can be reduced further.
\section*{Acknowledgments}
This research was supported by AcRF Tier $1$ grant T1 251RES1719 (Adaptive Data Structures: Concurrent, Cache-Efficient, Distributed). Peter Robinson acknowledges the support of the Natural Sciences and Engineering Research Council of Canada (NSERC).

\bibliographystyle{plain}
\bibliography{bibliography}

\begin{thebibliography}{10}

\bibitem{Afek:1985:TMB:323596.323613}
Yehuda Afek and Eli Gafni.
\newblock Time and message bounds for election in synchronous and asynchronous
  complete networks.
\newblock In {\em Proceedings of the Fourth Annual ACM Symposium on Principles
  of Distributed Computing}, PODC '85, pages 186--195, New York, USA, 1985.
  ACM.

\bibitem{Angluin:1980:LGP:800141.804655}
Dana Angluin.
\newblock Local and global properties in networks of processors (extended
  abstract).
\newblock In {\em Proceedings of the Twelfth Annual ACM Symposium on Theory of
  Computing}, STOC '80, pages 82--93, New York, NY, USA, 1980. ACM.

\bibitem{AW98}
Hagit Attiya and Jennifer Welch.
\newblock {\em Distributed Computing: Fundamentals, Simulations and Advanced
  Topics (2nd edition)}.
\newblock John Wiley Interscience, March 2004.

\bibitem{byz-leader}
John Augustine, Gopal Pandurangan, and Peter Robinson.
\newblock Fast byzantine leader election in dynamic networks.
\newblock In {\em Distributed Computing - 29th International Symposium, {DISC}
  2015, Tokyo, Japan, October 7-9, 2015, Proceedings}, pages 276--291, 2015.

\bibitem{Awerbuch:1987:ODA:28395.28421}
B.~Awerbuch.
\newblock Optimal distributed algorithms for minimum weight spanning tree,
  counting, leader election, and related problems.
\newblock In {\em Proceedings of the Nineteenth Annual ACM Symposium on Theory
  of Computing}, STOC '87, pages 230--240, New York, NY, USA, 1987. ACM.

\bibitem{DBLP:journals/jacm/AwerbuchGPV90}
Baruch Awerbuch, Oded Goldreich, David Peleg, and Ronen Vainish.
\newblock A trade-off between information and communication in broadcast
  protocols.
\newblock {\em J. {ACM}}, 37(2):238--256, 1990.

\bibitem{Bollobas}
B.~Bollob{\'a}s.
\newblock The isoperimetric number of random regular graphs.
\newblock {\em European Journal of Combinatorics}, 9(3):241 -- 244, 1988.

\bibitem{bollobas2001random}
B.~Bollob{\'a}s.
\newblock {\em Random Graphs}.
\newblock Cambridge Studies in Advanced Mathematics. Cambridge University
  Press, 2001.

\bibitem{Brunekreef1996}
Jacob Brunekreef, Joost-Pieter Katoen, Ron Koymans, and Sjouke Mauw.
\newblock Design and analysis of dynamic leader election protocols in broadcast
  networks.
\newblock {\em Distributed Computing}, 9(4):157, Feb 1996.

\bibitem{Chang:1979:IAD:359104.359108}
Ernest Chang and Rosemary Roberts.
\newblock An improved algorithm for decentralized extrema-finding in circular
  configurations of processes.
\newblock {\em Commun. ACM}, 22(5):281--283, May 1979.

\bibitem{Chatterjee:2018:CLE:3154273.3154308}
Soumyottam Chatterjee, Gopal Pandurangan, and Peter Robinson.
\newblock The complexity of leader election: A chasm at diameter two.
\newblock In {\em Proceedings of the 19th International Conference on
  Distributed Computing and Networking}, ICDCN '18, pages 13:1--13:10, NY, USA, 2018. ACM.

\bibitem{DOLEV1982245}
Danny Dolev, Maria Klawe, and Michael Rodeh.
\newblock An o(n log n) unidirectional distributed algorithm for extrema
  finding in a circle.
\newblock {\em Journal of Algorithms}, 3(3):245 -- 260, 1982.

\bibitem{815321}
C.~Fetzer and F.~Cristian.
\newblock A highly available local leader election service.
\newblock {\em IEEE Transactions on Software Engineering}, 25(5):603--618, Sep
  1999.

\bibitem{FLOCCHINI199676}
Paola Flocchini and Bernard Mans.
\newblock Optimal elections in labeled hypercubes.
\newblock {\em Journal of Parallel and Distributed Computing}, 33(1):76 -- 83,
  1996.

\bibitem{Frederickson:1987:ELS:7531.7919}
Greg~N. Frederickson and Nancy~A. Lynch.
\newblock Electing a leader in a synchronous ring.
\newblock {\em J. ACM}, 34(1):98--115, January 1987.

\bibitem{Gallager:1983:DAM:357195.357200}
R.~G. Gallager, P.~A. Humblet, and P.~M. Spira.
\newblock A distributed algorithm for minimum-weight spanning trees.
\newblock {\em ACM Trans. Program. Lang. Syst.}, 5(1):66--77, January 1983.

\bibitem{stacs2011_conductance}
George Giakkoupis.
\newblock Tight bounds for rumor spreading in graphs of a given conductance.
\newblock In {\em Proceedings of the 28th International Symposium on
  Theoretical Aspects of Computer Science (STACS)}, pages 57--68, March~10--12
  2011.

\bibitem{Gilbert:2010:DAO:1873601.1873679}
Seth Gilbert and Dariusz~R. Kowalski.
\newblock Distributed agreement with optimal communication complexity.
\newblock In {\em Proceedings of the Twenty-first Annual ACM-SIAM Symposium on
  Discrete Algorithms}, SODA '10, pages 965--977, Philadelphia, PA, USA, 2010.
  Society for Industrial and Applied Mathematics.

\bibitem{10.1007/3-540-40026-5_6}
Indranil Gupta, Robbert van Renesse, and Kenneth~P. Birman.
\newblock A probabilistically correct leader election protocol for large
  groups.
\newblock In Maurice Herlihy, editor, {\em Distributed Computing}, pages
  89--103, Berlin, Heidelberg, 2000. Springer Berlin Heidelberg.

\bibitem{hoory06}
Shlomo Hoory, Nathan Linial, and Avi Wigderson.
\newblock Expander graphs and their applications.
\newblock {\em Bull. Amer. Math. Soc.}, 43(04):439--562, August 2006.

\bibitem{Jerrum:1988:CRM:62212.62234}
Mark Jerrum and Alistair Sinclair.
\newblock Conductance and the rapid mixing property for markov chains: The
  approximation of permanent resolved.
\newblock In {\em Proceedings of the 20th Annual ACM Symposium on Theory of
  Computing}, STOC '88, pages 235--244, New York, USA, 1988. ACM.

\bibitem{Karp:2000:RRS:795666.796561}
R.~Karp, C.~Schindelhauer, S.~Shenker, and B.~Vocking.
\newblock Randomized rumor spreading.
\newblock In {\em Proceedings of the 41st Annual Symposium on Foundations of
  Computer Science}, FOCS '00, pages 565--, Washington, DC, USA, 2000. IEEE
  Computer Society.

\bibitem{doi:10.1137/0216019}
E.~Korach, S.~Moran, and S.~Zaks.
\newblock The optimality of distributive constructions of minimum weight and
  degree restricted spanning trees in a complete network of processors.
\newblock {\em SIAM Journal on Computing}, 16(2):231--236, 1987.

\bibitem{Kutten:2015:CUL:2742144.2699440}
Shay Kutten, Gopal Pandurangan, David Peleg, Peter Robinson, and Amitabh
  Trehan.
\newblock On the complexity of universal leader election.
\newblock {\em J. ACM}, 62(1):7:1--7:27, March 2015.

\bibitem{Kutten:2015:SBR:2945722.2945791}
Shay Kutten, Gopal Pandurangan, David Peleg, Peter Robinson, and Amitabh
  Trehan.
\newblock Sublinear bounds for randomized leader election.
\newblock {\em Theor. Comput. Sci.}, 561(PB):134--143, January 2015.

\bibitem{DBLP:conf/ifip/Lann77}
G{\'{e}}rard~Le Lann.
\newblock Distributed systems - towards a formal approach.
\newblock In {\em {IFIP} Congress}, pages 155--160, 1977.

\bibitem{Lynch:1996:DA:525656}
Nancy~A. Lynch.
\newblock {\em Distributed Algorithms}.
\newblock Morgan Kaufmann Publishers Inc., San Francisco, CA, USA, 1996.

\bibitem{Mitzenmacher:2005:PCR:1076315}
Michael Mitzenmacher and Eli Upfal.
\newblock {\em Probability and Computing: Randomized Algorithms and
  Probabilistic Analysis}.
\newblock Cambridge University Press, New York, NY, USA, 2005.

\bibitem{Molla:2017:DCM:3007748.3007784}
Anisur~Rahaman Molla and Gopal Pandurangan.
\newblock Distributed computation of mixing time.
\newblock In {\em Proceedings of the 18th International Conference on
  Distributed Computing and Networking}, ICDCN '17, pages 5:1--5:4, New York,
  NY, USA, 2017. ACM.

\bibitem{Motwani:1995:RA:211390}
Rajeev Motwani and Prabhakar Raghavan.
\newblock {\em Randomized Algorithms}.
\newblock Cambridge University Press, New York, NY, USA, 1995.

\bibitem{Pandurangan:2017:TMD:3055399.3055449}
Gopal Pandurangan, Peter Robinson, and Michele Scquizzato.
\newblock A time- and message-optimal distributed algorithm for minimum
  spanning trees.
\newblock In {\em Proceedings of the 49th Annual ACM SIGACT Symposium on Theory
  of Computing}, STOC 2017, pages 743--756, New York, NY, USA, 2017. ACM.

\bibitem{peleg}
D.~Peleg.
\newblock {\em Distributed Computing: A Locality-Sensitive Approach}.
\newblock Society for Industrial and Applied Mathematics, 2000.

\bibitem{Peleg:1990:TLE:78028.78040}
David Peleg.
\newblock Time-optimal leader election in general networks.
\newblock {\em J. Parallel Distrib. Comput.}, 8(1):96--99, January 1990.

\bibitem{Ramanathan2007}
Murali~K. Ramanathan, Ronaldo~A. Ferreira, Suresh Jagannathan, Ananth
  Grama, and Wojciech Szpankowski.
\newblock Randomized leader election.
\newblock {\em Distributed Computing}, 19(5):403--418, Apr 2007.

\bibitem{Ratnasamy:2001:SCN:964723.383072}
Sylvia Ratnasamy, Paul Francis, Mark Handley, Richard Karp, and Scott Shenker.
\newblock A scalable content-addressable network.
\newblock {\em SIGCOMM Comput. Commun. Rev.}, 31(4):161--172, 2001.

\bibitem{Rowstron:2001:PSD:646591.697650}
Antony I.~T. Rowstron and Peter Druschel.
\newblock Pastry: Scalable, decentralized object location, and routing for
  large-scale peer-to-peer systems.
\newblock In {\em Proceedings of the IFIP/ACM International Conference on
  Distributed Systems Platforms Heidelberg}, Middleware '01, pages 329--350,
  London, UK, UK, 2001. Springer-Verlag.

\bibitem{Sinclair}
A.~Sinclair.
\newblock {\em Algorithms for Random Generation and Counting}.
\newblock Birkhauser, Boston, USA, 1993.

\bibitem{491576}
G.~Singh.
\newblock Leader election in the presence of link failures.
\newblock {\em IEEE Transactions on Parallel and Distributed Systems},
  7(3):231--236, Mar 1996.

\bibitem{Vitanyi:1984:DEA:800057.808725}
Paul~M.B. Vit\'{a}nyi.
\newblock Distributed elections in an archimedean ring of processors.
\newblock In {\em Proceedings of the Sixteenth Annual ACM Symposium on Theory
  of Computing}, STOC '84, pages 542--547, New York, NY, USA, 1984. ACM.

\bibitem{Zhao:2006:TRG:2312202.2314945}
B.~Y. Zhao, Ling Huang, J.~Stribling, S.~C. Rhea, A.~D. Joseph, and J.~D.
  Kubiatowicz.
\newblock Tapestry: A resilient global-scale overlay for service deployment.
\newblock {\em IEEE J.Sel. A. Commun.}, 22(1):41--53, September 2006.

\end{thebibliography}

\end{document}